\documentclass[11pt]{article}
\usepackage{amsmath}
\usepackage{amsfonts}
\usepackage{amssymb}
\usepackage{amsthm}
\usepackage{thmtools, thm-restate}
\usepackage{fullpage}
\usepackage{comment}
\usepackage{graphicx}
\usepackage{hyperref}
\usepackage{makecell}
\usepackage{bbm}

\usepackage[
backend=biber,
style=alphabetic,
sorting=nyt,
maxbibnames=99 
]{biblatex}

\DeclareMathOperator{\dist}{dist}
\DeclareMathOperator{\supp}{supp}
\DeclareMathOperator{\wt}{wt}
\DeclareMathOperator{\rank}{rank}

\newcommand{\ac}{\mathsf{AC}^0}
\newcommand{\superconcentrator}{superconcentrator-induced }

\declaretheorem[{style=definition,numberwithin=section}]{definition}
\declaretheorem[{style=definition,sibling=definition}]{theorem}
\declaretheorem[{style=definition,sibling=definition}]{lemma}
\declaretheorem[{style=definition,sibling=definition}]{claim}
\declaretheorem[{style=definition,sibling=definition}]{corollary}
\declaretheorem[{style=definition,sibling=definition}]{proposition}

\title{On the Minimum Depth of Circuits with Linear Number of Wires Encoding Good Codes}
\author{Andrew Drucker\footnote{Independent. Email:\href{mailto:andy.drucker@gmail.com}{andy.drucker@gmail.com}.}  \and Yuan Li\footnote{School of Computer Science, Fudan University.  Email: \href{mailto:yuan_li@fudan.edu.cn}{yuan\_li@fudan.edu.cn}. An earlier version of this work appeared as Chapter 2 of the author's Ph.D. thesis \cite{Li17}.  Part of this work was done while the authors were affiliated with the University of Chicago Computer Science Dept.}}

\date{}

\begin{document}

\maketitle

\begin{abstract}

Let $S_d(n)$ denote the minimum number of wires of a depth-$d$ (unbounded fan-in) circuit encoding an error-correcting code $C:\{0, 1\}^n \to \{0, 1\}^{32n}$ with distance at least $4n$. G\'{a}l, Hansen, Kouck\'{y}, Pudl\'{a}k, and Viola [IEEE Trans. Inform. Theory 59(10), 2013] proved that $S_d(n) = \Theta_d(\lambda_d(n)\cdot n)$ for any fixed $d \ge 3$. By improving their construction and analysis, we prove $S_d(n)= O(\lambda_d(n)\cdot n)$. Letting $d = \alpha(n)$, a version of the inverse Ackermann function, we obtain circuits of linear size. This depth $\alpha(n)$ is the minimum possible to within an additive constant 2; we credit the nearly-matching depth lower bound to G\'{a}l \emph{et al.}, since it directly follows their method (although not explicitly claimed or fully verified in that work), and is obtained by making some constants explicit in a graph-theoretic lemma of  Pudl\'{a}k [Combinatorica, 14(2), 1994], extending it to super-constant depths.


We also study a subclass of MDS codes $C: \mathbb{F}^n \to \mathbb{F}^m$ characterized by the Hamming-distance relation $\dist(C(x), C(y)) \ge m - \dist(x, y) + 1$ for any distinct $x, y \in \mathbb{F}^n$.  (For linear codes this is equivalent to the generator matrix being totally invertible.)  We call these \emph{superconcentrator-induced codes}, and we show their tight connection with superconcentrators. Specifically, we observe that any linear or nonlinear circuit encoding a superconcentrator-induced code must be a superconcentrator graph, and any superconcentrator graph can be converted to a linear circuit, over a sufficiently large field (exponential in the size of the graph), encoding a superconcentrator-induced code.

\vspace{0.1cm}
\end{abstract}

\section{Introduction}
Understanding the computational complexity of \emph{encoding} error-correcting codes is an important task in theoretical computer science. Complexity measures of interest include time, space, and parallelism. Error-correcting codes are indispensable as a tool in computer science. Highly efficient encoding algorithms (or circuits) are desirable in settings studied by theorists including zero-knowledge proofs \cite{GLS+21}, circuit lower bounds \cite{CT19}, data structures for error-correcting codes \cite{Vio19}, pairwise-independent hashing \cite{Ish08}, and secret sharing \cite{Dru14}. Besides that, the \emph{existence} of  error-correcting codes with efficient encoding circuits sheds light on the designing of practical error-correcting codes.

We consider codes with constant rate and constant relative distance, which are called asymptotically good error-correcting codes or \emph{good codes} for short. The complexity of encoding good codes has been studied before. Bazzi and Mitter \cite{BM05} proved that branching programs with linear time and sublinear space cannot encode good codes. By using the sensitivity bounds \cite{Bop97}, one can prove that $\ac$ circuits cannot encode good codes; Lovett and Viola proved that $\ac$ circuits cannot sample good codes \cite{LV11}; Beck, Impagliazzo and Lovett \cite{BIL12} strengthened the result.

Dobrushin, Gelfand and Pinsker \cite{DGP73} proved that there exist linear-size circuits encoding good codes. Sipser and Spielman \cite{Spi96, SS96} explicitly constructed good codes that are encodable by bounded fanin circuits of depth $O(\log n)$ and size $O(n)$, and decodable by circuits of size $O(n \log n)$. For bounded fan-in, the depth $O(\log n)$ is obviously optimal. Henceforth, unless otherwise stated, we consider circuits with unbounded fan-in, where the size is measured by the number of \emph{wires} instead of gates.

G\'{a}l, Hansen, Kouck\'{y}, Pudl\'{a}k, and Viola \cite{GHK+12} investigated the circuit complexity of encoding good codes. G\'{a}l \emph{et al.} constructed circuits recursively and probabilistically, with clever recursive composition ideas, which resemble the construction of superconcentrators in \cite{DDP+83}. They also proved size lower bounds for bounded depth, by showing that any circuit encoding good codes must satisfy some superconcentrator-like properties; the lower bound follows from the size bounds for a variant of bounded-depth
superconcentrators studied by Pudl\'{a}k \cite{Pud94}.  Their construction's wire upper bounds are of form $O_d(n \cdot \lambda_d(n))$  (in our notation\footnote{Our definition of $\lambda_d(n)$ follows Raz and Shpilka \cite{RS03}. It is slightly different from G\'{a}l \emph{et al.}'s. In \cite{GHK+12}, the function $\lambda_i(n)$ is actually $\lambda_{2i}(n)$ in our notation.}) and their lower bounds are of form $\Omega_d(n \cdot \lambda_d(n))$, matching up to a multiplicative constant $c_d$ for constant values $d$. They also proved that there exist $O_d(n)$-size $O(\log \lambda_d(n))$-depth circuits encoding good codes. Here $\lambda_d(n)$ are slowly growing inverse Ackermann-type functions, e.g., $\lambda_2(n) = \Theta(\log n)$, $\lambda_3(n) = \Theta(\log\log n)$, $\lambda_4(n) = \Theta(\log^* n)$.

Druk and Ishai \cite{Dru14} proposed a randomized construction of good codes meeting the Gilbert-Varshamov bound, which can be encoded by linear-size logarithmic-depth circuits (with bounded fan-in). Their construction is based on linear-time computable pairwise independent hash functions \cite{Ish08}.

Chen and Tell \cite{CT19} constructed \emph{explicit} circuits of depth $d$ encoding linear codes with constant relative distance and code rate $\Omega\left(\frac{1}{\log n}\right)$ using $n^{1 + 2^{-\Omega(d)}}$ wires, for every $d \ge 4$. They used these explicit circuits to prove bootstrapping results for threshold circuits.

\subsection{Background and results}

To encode good error-correcting codes, linear size is obviously required. It is natural to ask, what is the minimum depth required to encode goods using a linear number of wires?  This question is addressed, but not fully answered, by the work of G\'{a}l \emph{et al.} \cite{GHK+12}. 

We show that one can non-explicitly encode error-correcting codes with constant rate and constant relative distance using $O(n)$ wires, with depth at most $\alpha(n)$, for sufficiently large $n$. Here, $\alpha(n)$ is a version of the inverse Ackermann function.  This is nearly optimal, by a lower bound of $\alpha(n) - 2$ that we credit to G\'{a}l \emph{et al.} \cite{GHK+12} as discussed below.  Our new upper bound states:

\begin{theorem}
\label{thm:depth_ub}
(Upper bound) Let $r \in (0, 1)$ and $\delta \in (0, \frac{1}{2})$ such that $r < 1 - h(\delta)$. For sufficiently large $n$ and for any $d \ge 3$, not necessarily a constant, there exists a linear circuit $C: \{0, 1\}^n \to \{0, 1\}^{\lfloor \frac{n}{r} \rfloor}$ of size $O_{r, \delta}(\lambda_d(n) \cdot n)$ and depth $d$ that encodes an error-correcting code with relative distance $\ge \delta$. In particular, when $d = \alpha(n)$, the circuit is of size $O_{r, \delta}(n)$.
\end{theorem}

Our upper bound improves the construction and the analysis by G\'{a}l \emph{et al.} \cite{GHK+12}. Specifically, let $S_d(n)$ denote the minimum size of a depth-$d$ linear circuit encoding a code $C:\{0, 1\}^n \to \{0, 1\}^{32n}$ with distance $4n$. G\'{a}l \emph{et al.} proved that $S_d(n) = O_d(\lambda_d(n)\cdot n)$, where the hidden constant in $O_d(\lambda_d(n)\cdot n)$ grows rapidly as an Ackermann-type function (if one follows their arguments and expand the computation straightforwardly in Lemma 26 in \cite{GHK+12}.) They also proved that, for any \emph{fixed} $m$, when the depth $d = O(\log(\lambda_m(n)))$, $S_d(n) = O_m(n)$. Their upper bound is strong, but suboptimal. Our main technical contribution is to prove $S_d(n) = O(\lambda_d(n) \cdot n)$, where the hidden constant is an absolute constant. Explaining how this improvement is obtained is a bit technical, and defeered to the next subsection.

Turning to the lower bounds: we credit the lower bound in the result below to G\'{a}l \emph{et al.}
(although it was not explicitly claimed or fully verified in that work), since it is directly obtainable by their size lower-bound method and the tool of  Pudl\'{a}k \cite{Pud94} on which it relies, when that tool is straightforwardly extended to super-constant depth.\footnote{The $\Omega(\cdot)$ notation in the circuit size lower bound of G\'{a}l \emph{et al.}, for example, Theorem 1 in \cite{GHK+12}, involves an implicit constant which decays with the depth $d$, as can be suitable for constant depths; similarly for the tool of Pudl\'{a}k, Theorem 3.(ii) in \cite{Pud94}, on which it relies. For general super-constant depths, more explicit work is required to verify the decay is not too rapid.}

\begin{theorem}
\label{thm:depth_lb}
(Lower bound) \cite{GHK+12}
Let $\rho \in (0, 1)$ and $\delta \in (0, \frac{1}{2})$, and let constant $c > 0$. Let $C_n : \{0, 1\}^{n} \to \{0, 1\}^{\lfloor n/\rho \rfloor}$ be a family of circuits of size at most $cn$ that encode error-correcting codes with relative distance $\ge \delta$.  Arbitrary Boolean-function gates of unrestricted fanin are allowed in $C_n$. If $n$ is sufficiently large, i.e., $n \ge N(r, \delta, c)$, the depth of the circuit $C_n$ is at least $\alpha(n) - 2$.
\end{theorem}

The proof for Theorem \ref{thm:depth_lb} closely follows \cite{GHK+12} and is an application of a graph-theoretic argument in the spirit of \cite{Val77, DDP+83, Pud94, RS03}. In detail, we use Pudl\'{a}k's size lower bounds \cite{Pud94} on ``densely regular'' graphs, and rely on the connection between good codes and densely regular graphs by G\'{a}l \emph{et al.} \cite{GHK+12}. Pudl\'{a}k's bound was originally proved for bounded depth; to apply it to unbounded depth, we explicitly determine the hidden constants by directly following Pudl\'{a}k's work, and verify that their decay at higher unbounded depths is moderate enough to allow the lower-bound method to give superlinear bounds up to depth $\alpha(n) - 3$. Even after our work, the precise asymptotic complexity of encoding good codes remains an open question for $d$ in the range [$\omega(1)$, $\alpha(n)$ - 3].

Stepping back to a higher-level view, the strategy of the graph-theoretic lower-bound arguments in the cited and related works is as follows:
\begin{itemize}
	\item Prove any circuit computing the target function must satisfy some superconcentrator-like connection properties;
	\item Prove any graph satisfying the above connection properties must have many edges;
	\item Therefore, the circuit must have many wires.
\end{itemize}

Valiant \cite{Val75, Val76, Val77} first articulated this kind of argument, and proposed the definition of superconcentrators. Somewhat surprisingly, Valiant showed that linear-size superconcentrators exist. As a result, one cannot prove superlinear size bounds using this argument (when the depth is unbounded). Dolev, Dwork, Pippenger, and Wigderson \cite{DDP+83} proved $\Omega(\lambda_d(n) \cdot n)$ lower bounds for bounded-depth (weak) superconcentrators, which implies circuit lower bounds for functions satisfying weak-superconcentrator properties. Pudl\'{a}k \cite{Pud94} generalized Dolev \emph{et al.}'s lower bounds by proposing the definition of \emph{densely regular graphs}, and proved lower bounds for bounded-depth densely regular graphs, which implies circuit lower bounds for functions satisfying densely regular property, including shifters, parity shifters, and Vandermonde matrices. Raz and Shpilka \cite{RS03} strengthened the aforementioned superconcentrator lower bounds by proving a powerful graph-theoretic lemma, and applied it to prove superlinear lower bounds for matrix multiplication. (This powerful lemma can reprove all the above lower bounds.) G\'{a}l \emph{et al.} \cite{GHK+12} proved that any circuits encoding good error-correcting codes must be densely regular. They combined this with Pudl\'{a}k's lower bound on densely regular graphs \cite{Pud94} to obtain $\Omega(\lambda_d(n) \cdot n)$ size bounds for depth-$d$ circuits encoding good codes.

All the circuit lower bounds mentioned above apply even to the powerful model of \emph{arbitrary-gate} circuits, that is,
\begin{itemize}
	\item each gate has unbounded fanin,
	\item a gate with fanin $s$ can compute any function from $\{0, 1\}^s$ to $\{0, 1\}$,
	\item circuit size is measured as the number of wires.
\end{itemize}
In this ``arbitrary-gates'' model, any function from $\{0, 1\}^n$ to $\{0, 1\}^m$ can be computed by a circuit of size $mn$.

It is known that any circuits encoding good codes must satisfy some superconcentrator-like connection properties \cite{Spi96}, \cite{GHK+12}. Our other result is a theorem in the \emph{reverse} direction in the algebraic setting over large finite fields. Motivated by this connection, we study \emph{\superconcentrator codes} (Definition \ref{def:sc_codes}), a subclass of maximum distance separable (MDS) codes \cite{LC04, GRS12}, and observe its tight connection with superconcentrators.

\begin{theorem}
\label{thm:sc_code_informal}
(Informal) Given any $(n, m)$-superconcentrator, one can convert it to a linear arithmetic circuit encoding a code $C: \mathbb{F}^n \to \mathbb{F}^m$ such that
\begin{equation}
\label{equ:sc_induced_def}
\dist(C(x), C(y)) \,\ge\, m - \dist(x, y) + 1 \quad \forall x \not= y \in \mathbb{F}^n
\end{equation}
by replacing each vertex with an addition gate and assigning the coefficient for each edge uniformly at random over a sufficiently large finite field (where $d2^{\Omega(n+m)}$ suffices, and $d$ is the depth of the superconcentrator).

\end{theorem}

We also observe that any arithmetic circuit, linear or nonlinear, encoding a code $C: \mathbb{F}^n \to \mathbb{F}^m$ satisfying \eqref{equ:sc_induced_def}, viewed as a graph, must be a superconcentrator.

The proof of Theorem \ref{thm:sc_code_informal} relates the connectivity properties with the rank of a matrix, and uses the Schwartz-Zippel lemma to estimate the rank of certain submatrices; these techniques are widely used, for example, in \cite{CCL13, Lov18}. In addition, the idea of assigning uniform random coefficients (in a finite field) to edges, to form linear circuits, has appeared before in e.g. network coding \cite{ACL+00, LYC03}.
The question we study is akin to a higher-depth version of the GM-MDS type questions about matrices with restricted support \cite{Lov18}.

Observe that any code satisfying the distance inequality \eqref{equ:sc_induced_def} is a good code. The existence of depth-$d$ size-$O(\lambda_d(n) \cdot n)$ superconcentrators \cite{DDP+83, AP94}, for any $d \ge 3$, immediately implies the existence of depth-$d$ (linear) arithmetic circuits of size $O(\lambda_d(n) \cdot n)$ encoding good codes \emph{over large finite field}.

In a subsequent work \cite{Li23}, inspired by this connection and using similar techniques, the second author proved that any $(n, m)$-superconcentrator can compute the shares of an $(n, m)$ linear threshold secret sharing scheme. In other words, any $(n, m)$-\superconcentrator code induces an $(n, m)$ linear threshold secret sharing scheme. Results in \cite{Li23} can be viewed as an application of \superconcentrator codes.

\subsection{Circuit construction techniques}
In terms of techniques, G\'{a}l \emph{et al.} proposed the notion of \emph{range detectors} (Definition \ref{def:range_detector}). They use probabilistic methods to prove certain linear-size depth-1 range detectors \emph{exist} and then compose the range detectors \emph{recursively} to construct circuits encoding good codes. We improve the recursive construction by tuning the parameters and changing the way range detectors are composed (in the inductive step), which leads to an \emph{absolute constant} in the size bound $O(\lambda_d(n)\cdot n)$, eliminating the dependency on $d$. By utilizing a property of the inverse Ackermann function, we immediately obtain circuits of depth $\alpha(n)$ and size $O(n)$, avoiding a size-reduction argument.

Let us give an overview of the construction, highlighting the difference between G\'{a}l \emph{et al.}'s. First, we distill some notions to facilitate the construction, each of which is either explicit or present in essence in \cite{GHK+12}.
\begin{itemize}
    \item Partial good code (i.e., Definition \ref{def:pgc}), or PGC for short: when the weight of an input is within a certain range, a PGC outputs a codeword whose relative weight is above an absolute constant.
    \item Condenser: a condenser function reduces the number of inputs while preserving the \emph{minimum} absolute weight. Using probabilistic arguments, G\'{a}l \emph{et al.} proved that certain \emph{unbalanced unique-neighbor expanders} exist, and converted the expanders to linear-size condensers, that is, $\left(n, \lfloor \frac{n}{r} \rfloor, s, \frac{n}{r^{1.5}}, s\right)$-range detectors. We reuse their condensers.
    \item Amplifier: these amplify the number of inputs while preserving the \emph{relative distance}.  G\'{a}l \emph{et al.} uses a probabilistic argument and applies the Chernoff bound to show linear-size amplifiers exist (e.g., in the proof of Theorem 22 and Lemma 26). We choose to decouple the argument by using the disperser graphs given by \cite{RT00}, choosing the coefficients at random, and then applying a union bound. 

    \item Rate booster: these do not change the number of inputs, but can boost the rate close to the Gilbert-Varshamov bound. Lemma 17 in \cite{GHK+12} embodies a rate booster, although a not fully optimized one; G\'{a}l \emph{et al.} also state\footnote{In the conference version (COCOON 2023), we mistakenly stated that such rate-boosting is not discussed by G\'{a}l \emph{et al.} (based on non-journal versions of the paper), but in fact it is discussed in Section 4 of the journal version.} that a tighter analysis would approach the Gilbert-Varshamov bound. Using a probabilistic argument similar to that with the amplifiers and to the argument in \cite{GHK+12}, we provide a verification of this claim.
    \item Composition Lemma: this tool allows us to combine $h$ partial good codes into a larger one, at the cost of increasing the depth by 1, and the size by $O(n)$. Furthermore, if each partial good code has bounded output fanin $O(1)$, then one can \emph{collapse} the last layer at the cost of increasing the size by $O(hn)$. The composition idea was implicit in G\'{a}l \emph{et al.}'s (e.g., in Theorem 22 and Lemma 26 in \cite{GHK+12}).
\end{itemize}

With the these concepts, the task is to construct the circuits using these building blocks. Our analysis framework is also inspired by \cite{DDP+83}; we show the minimal size of a depth-$d$ $(n, r, s)$-PGC satisfies a similar recursion as that of \emph{partial superconcentrators}. For the base case when the depth is small (i.e., $d \le 4$), we completely reuse G\'{a}l \emph{et al.}'s construction. In the key inductive step, G\'{a}l \emph{et al.} construct a depth-$(d+2)$ circuits encoding an $\left(n, \frac{n}{r}, n\right)$-PGC by composing a number of, say $h$, PGCs. \emph{Assuming} each smaller PGC is of size $O(n)$, the total size would be $O(hn)$. They let
\begin{equation}
\label{equ:Gal_param}
    \begin{cases}
      k_1 & =\,  \min\{r, \lceil m^{3/4} \rceil \}\\
      k_{i+1} & =\, \lambda_d(k_i)^3,
    \end{cases}       
\end{equation}
where $h$ is the least integer such that $k_h$ is below a particular universal constant. Each inner part is an $\left(\left\lfloor\frac{n}{\lambda_{d}(k_i)^2} \right\rfloor, \frac{n}{k_i}, \frac{n}{\lambda_{d}(k_i)^2} \right)$-PGC of depth $d$ and size $O_d(n)$, topped by a condenser, appended by an ampilifier. Then they argue that $h = 2\lambda_{d+2}(n) + O_d(1)$, and thus the total size of the depth-$(d+2)$ circuit becomes $O_{d+2}(\lambda_{d+2}(n)\cdot n)$. They did not mention the constant in $O_d(1)$; if one works out the constant straightforwardly, it is a very fast-growing constant depending on $d$.

Our remedy is to choose the parameters and compose the PGCs differently in the inductive construction. To obtain a depth-$(d+2)$ circuit encoding an $\left(n, \frac{n}{r}, n\right)$-PGC, we compose at most $\lambda_{d+2}(n)$ PGCs that are of the parameter $\left(n, \frac{n}{A^{(i)}_{k-1}(c_0)}, \frac{n}{A^{(i-1)}_{k-1}(c_0)}\right)$, and of size $O(n)$, where $1 \le i \le \lambda_{d+2}(n)$. In addition, we ensure that each PGC has bounded output fanin $O(1)$, which allows us to collapse the last layer in the composition lemma without blowing up the size. Let $r = A^{(i-1)}_{k-1}(c_0)$. To construct a depth-$(d+2)$ $\left(n, \frac{n}{A_{k-1}(r)}, \frac{n}{r}\right)$-PGC, we cannot reduce it to a depth-$d$ PGC directly, due to the parameter restriction in the condenser. (Recall that the linear-size condenser constructed is an $\left(n, \lfloor \frac{n}{r} \rfloor, s, \frac{n}{r^{1.5}}, s\right)$-range detector, not an $\left(n, \lfloor \frac{n}{r} \rfloor, s, \frac{n}{r}, s\right)$-range detector.) We work around this by composing two PGCs: an $\left(n, \frac{n}{A(k-1, r)}, \frac{n}{4r^2}\right)$-PGC, and an $\left(n, \frac{n}{4r^2}, \frac{n}{r}\right)$-PGC, where each is of linear size, and have bounded output fanin $O(1)$. The former can be constructed by a linear-size condenser, an $\left(\left\lfloor \frac{n}{2r} \right\rfloor, \frac{n}{A(k-1, r)}, \frac{n}{2r}\right)$-PGC of depth $d$, and an amplifier; the latter we construct directly by a simple argument. This composition is inspired by the superconcentrator construction by Dolev \emph{et al.} \cite{DDP+83}, who constructed superconcentrators of depth $d$ and size $O(\lambda_d(n) \cdot n)$. Indeed, the encoding circuits \emph{have to} resemble the superconcentrators, since it was known that the circuits must be a (weak) superconcentrators.

The above sketch leads to a construction of depth-$d$ size-$O(\lambda_d(n) \cdot n)$ circuits encoding good codes, where $d$ is not necessarily a constant; by letting $d = \alpha(n)$, one immediately gets depth-$\alpha(n)$ size-$O(n)$ circuits. This improves the analysis, getting rid of an additive $O(\log \alpha(n))$ factor in the depth (compared with Corollary 32 in \cite{GHK+12}, or Corollary 1.1 in \cite{DDP+83}). But a tight analysis to prove this requires one small new ingredient compared to prior works. Specifically, we observe that, if $\lambda_d(n) \le d$, then $\lambda_{d+2}(n) = O(1)$ (See Proposition A.2 (ii) for details.) The above observation can also improve the depth bound on the linear-size superconcentrators in \cite{DDP+83} from $\alpha(n) + O(\log \alpha(n))$ to $\alpha(n)$ (without change in the construction). In contrast, G\'{a}l \emph{et al.} deployed a different strategy to achieve linear size. Having constructed an encoding circuit of depth $d = O(1)$ and size $O_d(\lambda_d(n)\cdot n)$, they then apply a size-reduction transformation $O(\log \lambda_d(n))$ times, a transformation which reduces a size-$m$ instance to a size $m/2$ instance with additional cost $O(m)$ wires. In this way they obtain a circuit of depth $2d+O(\log \lambda_d(n))$ and size $O_d(n)$.

We verify our codes can obtain any constant rate and constant relative distance within the Gilbert-Varshamov bound, by using a disperser graph at the bottom layer and collapsing that layer afterward (for linear circuits). A similar rate-boosting by probabilistic construction is described in \cite{GHK+12}; see their (quantitatively somewhat weaker) Lemma 17 and Corollary 18 and the comments following, stating an improved analysis is possible.

\section{Inverse Ackermann functions}
\begin{definition}
\label{def:lambda}
(Definition 2.3 in \cite{RS03}) For a function $f$, define $f^{(i)}$ to be the composition of $f$ with itself
$i$ times. For a function $f:\mathbb{N} \to \mathbb{N}$ such that $f(n) < n$ for all $n > 0$, define
\[
f^*(n) \,=\, \min \{ i : f^{(i)}(n) \le 1\}.
\]

 Let
\begin{eqnarray*}
\lambda_1(n) & \,=\, & \lfloor \sqrt{n} \rfloor \ , \\
\lambda_2(n) & \,=\, & \lceil \log n \rceil \ , \\
\lambda_d(n) & \,=\, & \lambda^*_{d-2}(n) \ .
\end{eqnarray*}
\end{definition}

As $d$ gets larger, $\lambda_d(n)$ becomes extremely slowly growing, for example, $\lambda_3(n) = \Theta(\log\log n)$, $\lambda_4(n) = \Theta(\log^* n)$, $\lambda_5(n) = \Theta(\log^* n)$, etc.

We define a version of the inverse Ackermann function as follows.

\begin{definition} [Inverse Ackermann Function] For any positive integer $n$, let
\[
\alpha(n) \,=\, \min\{\text{even } d: \lambda_d(n) \le 6 \}.
\]
\end{definition}

There are different variants of the inverse Ackermann function; they differ by at most a multiplicative constant factor.

We need the definition of the Ackermann function. We put some relevant properties in Appendix \ref{app:inv_Ack_properites}.

\begin{definition}
\label{def:acker}
(Ackermann function \cite{Tar75, DDP+83}) Define
\begin{equation}
\begin{cases}
	A(0, j) \,=\, 2j, & \text{for $j \ge 1$}\\
	A(i, 1) \,=\, 2, & \text{for $i \ge 1$}\\
	A(i, j) \,=\, A(i-1, A(i, j-1)), & \text{for $i \ge 1, j \ge 2$}.
\end{cases}
\end{equation}
\end{definition}

The Ackermann function grows rapidly. From the definition, one can easily verify that $A(0, i) = 2i$, $A(1, i) = 2^i$ and
$
A(2, i)=
\begin{matrix}
\underbrace{2^{2^{{}^{.\,^{.\,^{.\,^2}}}}}} & 
\\  
i\mbox{ copies of }2
    &
\end{matrix}.
$
For notational convenience, we often write $A(i, j)$ as $A_i(j)$.

\section{Upper bound}

In this section, we prove Theorem \ref{thm:depth_ub}. That is, we non-explicitly construct, for any rate $r \in (0, 1)$ and relative distance $\delta \in (0, \frac{1}{2})$ satisfying $r < 1 - h(\delta)$, circuits encoding error-correcting codes $C:\{0, 1\}^n \to \{0, 1\}^{\lfloor n/r \rfloor}$ with relative distance $\delta$, where the circuit is of size $O_{r, \delta}(\lambda_d(n) \cdot n)$ and depth $d$. In particular, when the depth $d = \alpha(n)$, the circuit is of size $O_{r, \delta}(n)$.

First, we construct circuits encoding codes $C:\{0, 1\}^n \to \{0, 1\}^{32n}$ with relative distance $\frac{1}{8}$, where the constants $32$ and $\frac{1}{8}$ are picked for convenience. Then, we verify that one can boost the rate and the distance to achieve the Gilbert-Varshamov bound (without increasing the depth of the circuit).

Note that random linear codes achieve the Gilbert-Varshamov bound. However, circuits encoding random linear codes have size $O(n^2)$. In contrast, our circuits have size $O(n)$. Our circuits consist of XOR gates only; we call these \emph{linear circuits} hereafter. We point out that the construction generalizes to any finite field, where XOR gates are replaced by addition gates (over that finite field).

Recall that let $S_d(n)$ denote the minimum size of a depth-$d$ linear circuit encoding a code $C:\{0, 1\}^n \to \{0, 1\}^{32n}$ with distance $4n$.

\begin{definition} \label{def:pgc} (Partial good code) $C: \{0, 1\}^n \to \{0, 1\}^{32n}$ is called \emph{$(n, r, s)$-partial good code}, or \emph{$(n, r, s)$-PGC} for short, if for all $x \in \{0, 1\}^n$ with $\wt(x) \in [r, s]$, we have $\wt(C(x)) \ge 4n$.
\end{definition}

Denote by $S_d(n, r, s)$ the minimum size of a linear circuit of depth-$d$ that encodes an $(n, r, s)$-PGC, where $r, s$ are real numbers.

\begin{definition} \cite{Sip88, CW89} Bipartite graph $G = (V_1 = [n], V_2 = [m], E)$ is a \emph{$(k, \epsilon)$-disperser graph} if for any $X \subseteq V_1$ with $|X| \ge k$, $|\Gamma(X)| \ge (1 - \epsilon)m$.
\end{definition}

We rely on the following probabilistic construction of dispersers.

\begin{theorem}
\label{thm:prob_disperser}
(Theorem 1.10 in \cite{RT00}) For every $n \ge k > 1$, $m \ge 1$, and $\epsilon > 0$, there exists a $(k, \epsilon)$-disperser graph $G = (V_1 = [n], V_2 = [m], E)$ with left degree
\[
D \:=\: \left\lceil \frac{1}{\epsilon} \left(\ln\left(\frac{n}{k}\right)+1\right) + \frac{m}{k}\left(\ln\left(\frac{1}{\epsilon}\right) + 1\right) \right\rceil.
\]
	
\end{theorem}

A \emph{probabilistic (linear) circuit} is a linear circuit with random coefficients. That is, each XOR gate computes $c_1 x_1 + c_2 x_2 + \ldots + c_m x_m$, where $x_1, x_2, \ldots, x_m$ are incoming wires, coefficients $c_1, c_2, \ldots, c_n \in \{0, 1\}$ are chosen independently and uniformly at random. A probabilistic circuit can be viewed as a distribution of linear circuits. Our work in what follows relies on the following simple observation: if one of the incoming wires is nonzero, the output is uniformly distributed in $\{0, 1\}$.

\begin{figure}[ht]
\includegraphics[scale=0.4]{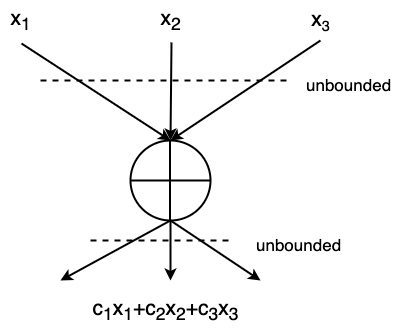}
\centering
\caption{Gate in a probabilistic circuit}
\label{fig:xor_gate}
\end{figure}

\begin{lemma}
\label{lem:rate_booster}	
(Rate booster) Let $\delta \in (0, 32)$, and let $c > 1$, $\gamma \in (0, \frac{1}{2})$ be such that $h(\gamma) < 1 - \frac{1}{c}$.	 For sufficiently large $n \ge N(c, \gamma)$, there exists a probabilistic circuit $C : \{0, 1\}^{32n} \to \{0, 1\}^{\lfloor cn \rfloor}$ of depth 1 and size $O_{\delta, c, \gamma}(n)$ such that
\[
\Pr_C\left[\wt(C(x)) \ge \gamma \cdot \lfloor cn \rfloor\right] \:>\: 1 - 2^{-n}
\]
for any $x \in \{0, 1\}^{32n}$ with $\wt(x) \ge \delta n$. Moreover, the output gates have bounded fanin $O_{\delta, c, \gamma}(1)$.
\end{lemma}

\begin{proof} Let $c' = \frac{\lfloor cn \rfloor}{n}$. Let $n$ be sufficiently large such that $h(\gamma) < 1 - \frac{1}{c'}$.

By Theorem \ref{thm:prob_disperser}, there exists a $(\delta n, \epsilon)$-disperser graph $G' = (V_1 = [32n], V'_2 = [(c'+1)n], E)$, where $\epsilon > 0$ is small constant to be determined. After removing $n$ vertices in $V'_2 = [(c'+1)n]$ with the largest degree, we are left with a bipartite graph $G = (V_1 = [32n], V_2=[c'n], E)$ such that
\begin{equation}
\label{equ:G_exp}	
|\Gamma(X)| \:\ge\: c'n - \epsilon (c'+1)n \:=\: c'n \left(1 -\epsilon - \frac{\epsilon}{c'}\right)
\end{equation}
for all $X \subseteq V_1$ of size at least $\delta n$. Moreover, the degree of the vertices in $V_2$ is bounded by $32 D$, where
$
D = \left\lceil \frac{1}{\epsilon} (\ln(\frac{32}{\delta})+1) + \frac{c'+1}{\delta}(\ln(\frac{1}{\epsilon}) + 1)\right\rceil = O_{\delta, c, \epsilon}(1).
$

We transform bipartite graph $G$ into a probabilistic circuit by identifying vertices in $V_1$ with $32n$ inputs, denoted by $x_1, x_2, \ldots, x_{32n}$, and replacing each vertex in $V_2$ by an XOR gate with random coefficients. That is, if vertex $j \in V_2$ is connected to $i_1, i_2, \ldots, i_k \in V_1$, then vertex $j$ computes an output
\[
y_j \:=\: c_{j, i_1} x_{i_1} + c_{j, i_2} x_{i_2} + \ldots + c_{j, i_k} x_{i_k} \pmod 2,
\]
where $c_{j, i_1}, c_{j, i_2}, \ldots, c_{j, i_k} \in \{0, 1\}$ are chosen independently and uniformly at random.

For any $x \in \{0, 1\}^{32n}$ with $\wt(x) \ge \delta n$, let $X = \supp(x)$. By the expansion property \eqref{equ:G_exp}, we have $|\Gamma(X)| \ge c'n(1-\epsilon')$, where $\epsilon' = (1 + \frac{1}{c'})\epsilon$. For every $j \in \Gamma(X)$, vertex $j$ is incident to at least one vertex in $X$; pick an arbitrary neighbor, denoted by $i \in X$, and leave coefficient $c_{i, j}$ unfixed. In other words, we fix all the coefficients except those $|\Gamma(X)|$ coefficients (that are incident to $X$). Observe that $C(x)$, restricted to $\Gamma(X)$, is uniformly distributed in $\{0, 1\}^{|\Gamma(X)|}$. Thus,
\begin{eqnarray}
\Pr_C[\wt(C(x)) < \gamma c'n] 
& \le & \frac{{\lceil c'n(1-\epsilon') \rceil \choose 0} + {\lceil c'n(1-\epsilon') \rceil \choose 1} + \ldots + {\lceil c'n(1-\epsilon') \rceil \choose \lfloor \gamma c'n \rfloor}}{2^{\lceil c'n(1-\epsilon') \rceil}} \nonumber \\
& \le & 2^{-c'n(1-\epsilon')(1 - h(\frac{\gamma}{1 - \epsilon'}))}. \label{equ:fail_prob}
\end{eqnarray}

Recall that $\epsilon' = (1 + \frac{1}{c'})\epsilon$. Let $\epsilon > 0$ be sufficiently small such that $c'(1-\epsilon')(1 - h(\frac{\gamma}{1 - \epsilon'})) > 1$, the probability \eqref{equ:fail_prob} would be less than $2^{-n}$. (Such $\epsilon' > 0$ clearly exists, since $h(\gamma) < 1 - \frac{1}{c'} \Leftrightarrow c'(1-h(\gamma)) > 1$).
\end{proof}

Next, a ``composition lemma'' allows us to combine several partial good codes into a larger one. The lemma is already implicit in \cite{GHK+12}.

\begin{lemma}
\label{lem:composition}
(Composition Lemma) For any $n, d \ge 1$, $\ell \ge 2$, and any $1 \le r_1 \le r_2 \le \ldots r_{\ell+1} \le n$, assume there exist linear circuits of size $s_i$ and depth $d$ encoding some $(n, r_{i}, r_{i+1})$-PGC, $i = 1, 2, \ldots, \ell$. We have
\begin{equation}
	S_{d+1}(n, r_1, r_{\ell+1}) \,\,\le\,\, \sum_{i = 1}^\ell s_i + O(n).
\end{equation}

Moreover, if the output gates of the linear circuits encoding $(n, r_{i}, r_{i+1})$-PGCs have bounded fan-in $D$, for all $i = 1, 2, \ldots, \ell$, then we have
\begin{equation}
\label{equ:composition_depthd}
	S_{d}(n, r_1, r_{\ell+1}) \,\,\le\,\, \sum_{i = 1}^\ell s_i + O(D\ell n).
\end{equation}
In addition, the fanin of the output gates (of the depth-$d$ circuit encoding the $(n, r_1, r_{\ell+1})$-PGC) is bounded by $O(\ell D)$. 
\end{lemma}
\begin{proof}
Denote the linear circuit encoding an $(n, r_i, r_{i+1})$-PGC by $C_i$, $i = 1, 2, \ldots, \ell$. Recall that each $C_i : \{0, 1\}^n \to \{0, 1\}^{32n}$ has $32n$ outputs, denoted by $w_1^{(i)}, w_2^{(i)}, \ldots, w_{32n}^{(i)}$.

Create $32n$ gates, denoted by $y_1, y_2, \ldots, y_{32n}$. Let 
\[
y_j \,\,=\,\, \sum_{i = 1}^\ell c_j^{(i)}w_j^{(i)} \pmod 2,
\]
where coefficients $c_j^{(i)} \in \{0, 1\}$ are chosen independently and uniformly at random. See Figure \ref{fig:composition} for illustration. 

\begin{figure}[ht]
\includegraphics[scale=0.4]{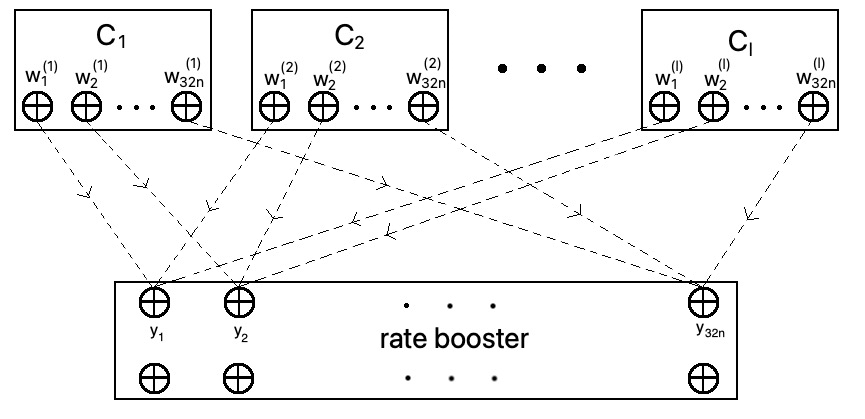}
\centering
\caption{Partial good codes composition}
\label{fig:composition}
\end{figure}

Figure \ref{fig:composition} is a bit misleading. Instead of creating $32n$ new gates $y_1, y_2, \ldots, y_{32n}$, we \emph{merge} $w_j^{(1)}, w_j^{(2)}, \ldots, w_j^{(\ell)}$, for each $j = 1, 2, \ldots, 32n$. Fixing $j$, independently for each $i \in \{1, 2, \ldots, \ell\}$, with probability $\frac{1}{2}$, merge $w_j^{(i)}$ with $y_j$, and with probability $\frac{1}{2}$, do nothing. As such, the depth of the circuit is still $d$, and the size is at most $\sum_{i = 1}^\ell s_i$.

Let us analyze the weight of $y = (y_1, y_2, \ldots, y_{32n}) \in \{0, 1\}^{32n}$. For any $x \in \{0, 1\}^n$ with $\wt(x) \in [r_1, r_{\ell+1}]$, there exists $k$ such that $\wt(x) \in [r_k, r_{k+1}]$. By the definition of PGC, $C_k(x)$ has weight at least $4n$. Fix all coefficients $c_j^{(i)}$ except for those $i = k$ and $j \in \supp(C_k(x))$, one can easily see that $\{y_j : j \in \supp(C_k(x))\}$ is uniformly distributed in $\{0, 1\}^{|\supp(C_k(x))|}$. Thus, for any fixed $x$, 
\begin{eqnarray*}
\Pr\left[\wt(y) < \frac{n}{2}\right] & \,\le\, & \frac{{4n \choose 0} + {4n \choose 1} + \ldots + {4n \choose \lfloor \frac{n}{2} \rfloor}}{2^{4n}} \\
	& \,\le\, & 2^{-4n(1 - h(\frac{1}{8}))} \,\,<\,\, 2^{-1.8n}.
\end{eqnarray*}
Taking a union bound over $x$, we know there exists an assignment of coefficients such that the weight of the output is always $\ge \frac{n}{2}$.

Finally, we boost the distance from $\frac{1}{2}n$  to $4n$ using a rate booster in Lemma \ref{lem:rate_booster}. Specifically, apply Lemma \ref{lem:rate_booster} with $\delta = \frac{1}{2}, c = 32, \gamma = \frac{1}{8}$. Adding the rate booster increases the depth by 1, and increases the size by $O(n)$.

For the ``moreover'' part, observe that the outputs in the rate booster have bounded fan-in $O(1)$. As in the use of rate boosting in \cite{GHK+12}, we collapse the last layer, where, originally, the last layer has bounded fan-in $O(1)$, and the second last layer has fan-in $\le \ell D$. As such, $O(n)$ wires in the rate booster will extend to $O(D\ell n)$ wires, which proves \eqref{equ:composition_depthd}. Furthermore, the output gates have fanin $O(\ell D)$. 
\end{proof}

\begin{definition} \label{def:range_detector} \cite{GHK+12} An \emph{$(m, n, \ell, k, r, s)$-range detector} is a linear circuit that has $m$ inputs, $n$ outputs, and on any input of Hamming weight between $\ell$ and $k$, it outputs a string with Hamming weight between $r$ and $s$. We omit the last parameter if $s = n$.
\end{definition}

\begin{lemma}
\label{lem:condenser}
 (Condenser, Lemma 23 in \cite{GHK+12}) \label{lem:range_detector} There exists a constant $c_0 \ge 6$ such that for any $c_0 \le r \le n$ and $1 \le s \le \frac{n}{r^{1.5}}$, where $n$ is an integer, $r, s$ are reals, there exists an
$
\left(n, \lfloor \frac{n}{r} \rfloor, s, \frac{n}{r^{1.5}}, s\right)\text{-range detector}
$
of depth 1 and size $6n$.
\end{lemma}

The following lemma says there exist linear size ``amplifiers'' that amplify the number of outputs while preserving the relative distance.

\begin{lemma} [Amplifier]
\label{lem:amplifier} For any positive integers $n, m \ge 3n$, there exists an $(n, m, \frac{n}{8}, n, \frac{m}{8})$-range detector computable by a depth-1 linear circuit of size $O(m)$. Moreover, the fanin of the outputs is bounded by an absolute constant.
\end{lemma}
\begin{proof}
By Theorem \ref{thm:prob_disperser}, there exists a bipartite graph 	$G=(V_1=[n], V_2=[m], E)$ with $O(m)$ edges such that for every $X \subseteq V_1$ of size at least $\frac{n}{8}$, $|\Gamma(X)| \ge \frac{9}{10} \cdot m$. In addition, the vertices in $V_2$ have bounded degrees. (For example, applying Theorem \ref{thm:prob_disperser} with $n = m$, $m = 2m$, $k = \frac{n}{8}$, and $\epsilon = \frac{1}{20}$, we have a disperser with left degree $O(\frac{m}{n})$ and right average degree $O(1)$. After removing $m$ right vertices with largest degree, we have the desired graph.)

Convert graph $G$ into a probabilistic circuit $C:\{0, 1\}^n \to \{0, 1\}^m$ as follows:
\begin{itemize}
	\item View the vertices in $V_1$ as inputs, denoted by $x_u$, for each $u \in V_1$.
	\item Replace the vertices in $V_2$ by XOR gates. View them as outputs, denoted by $y_v$, for each $v \in V_2$.
	\item For each edge $(u, v) \in E$, choose a coefficient $c_{u, v} \in \{0, 1\}$ independently at random. Let the output gate $v$ compute the function
$$
y_v \:=\: \sum_{(u, v) \in E} c_{u, v} x_u \pmod 2.
$$
\end{itemize}

For any $x \in \{0, 1\}^n$ with $\wt(x) \ge \frac{n}{8}$, by the expansion property of graph $G$, $|\Gamma(\supp(x))| \ge \frac{9}{10} \cdot  m$. In other words, there exist at least $\lceil \frac{9}{10} \cdot  m \rceil$ outputs, which are incident to at least one nonzero input. So, $C(x)$, restricted to $\Gamma(\supp(x))$, is uniformly distributed in $\{0, 1\}^{|\Gamma(\supp(x))|}$. Thus,
\begin{eqnarray*}
	\Pr[\wt(C(x)) < \frac{m}{8}] & \,\,\le\,\, & \frac{{\lceil \frac{9}{10} \cdot m \rceil \choose 0} + {\lceil\frac{9}{10} \cdot m \rceil \choose 1} + \ldots + {\lceil\frac{9}{10} \cdot m \rceil \choose \lfloor \frac{m}{8} \rfloor}}{2^{\lceil \frac{9}{10} \cdot m \rceil}} \\
	& \,\,\le\,\, & 2^{-(1 - h(\frac{5}{36})) \cdot \frac{9}{10} m} \\
	& \,\,<\,\, & 2^{-0.37 m} \,\,<\,\, 2^{-n}.
\end{eqnarray*}
Finally, taking a union bound over all $x$, we conclude there exists a linear circuit that computes an $(n, m, \frac{n}{8}, n, \frac{m}{8})$-range detector.
\end{proof}

The following lemma says, by putting a condenser at the top and an amplifier at the bottom, one can reduce the size of the problem (of encoding partial good codes).

\begin{lemma} (Reduction Lemma)
\label{lem:reduction} Let $c_0$ be the constant in Lemma \ref{lem:condenser}.
 For any $r \in [c_0, n]$ and $1 \le s \le t \le \frac{n}{r^{1.5}}$,
\begin{equation}
\label{equ:red_ineq}
S_d(n, s, t) \,\,\le\,\, S_{d-2}\left(\left\lfloor \frac{n}{r} \right\rfloor, s, \frac{n}{r}\right) + O(n).
\end{equation}
In addition, the output gates encoding $(n, s, t)$-PGC have bounded fanin $O(1)$.
\end{lemma}
\begin{proof} Construct a depth-$d$ linear circuit that encodes an $(n, s, t)$-PGC by concatenating three circuits:
\begin{itemize}
	\item (Condenser) The top is an $(n, \lfloor \frac{n}{r} \rfloor, s, t, s)$-range detector of size $O(n)$ and depth 1, whose existence is proved by Lemma \ref{lem:condenser}.
	\item (Partial good code) The middle is an $(\lfloor \frac{n}{r} \rfloor, s, \lfloor \frac{n}{r} \rfloor)$-PGC of depth $d-2$.
	\item (Amplifier) The bottom is an $(32 \lfloor \frac{n}{r} \rfloor, 32 n,  4 \lfloor \frac{n}{r} \rfloor, 32 \lfloor \frac{n}{r} \rfloor, 4n)$-range detector of depth 1 and size $O(n)$, whose existence is proved by Lemma \ref{lem:amplifier}. Moreover, the output gates have bounded fanin $O(1)$.
\end{itemize}
The output gates of the overall circuit have bounded fanin $O(1)$, because the output gates of the amplifier have bounded fanin $O(1)$.
\end{proof}

The following bounds are tight for fixed depth $d$. To avoid redoing the same work, we will use it for the case $d = 4$ at the base step in our main construction. Our technical contribution is to remove the dependence on $d$, i.e., we prove $S_d = O(\lambda_d(n)\cdot n)$ for any $d$, not necessarily fixed, where the hidden constant is an absolute constant. 

\begin{lemma}
\label{lem:ub_gal}
(Lemma 28 in \cite{GHK+12}) Fix an even constant $d \ge 4$. For any $1 \le r \le n$,
\[
S_d\left(n, \frac{n}{r}, n\right) \,\,=\,\, O_d(\lambda_d(r) \cdot n),
\]
where the constant in $O_d(\lambda_d(r) \cdot n)$ only depends on $d$.
\end{lemma}

\begin{lemma}
\label{lem:ub_d2}
(Lemma 27 in \cite{GHK+12}) For any $1 \le r \le n$,
\[
S_2\left(n, \frac{n}{r}, n\right) \,\,=\,\, O(n \log^2 r).
\]
\end{lemma}

We will need the following handy result in our recursive construction.

\begin{lemma}
\label{lem:handy}	
For any $1 \le r \le n$,
\[
S_4\left(n, \frac{n}{4r^2}, \frac{n}{r}\right) \,\,=\,\, O(n) \ ,
\]
where the constant in $O(n)$ is an absolute constant. Moreover, the output gates of the linear circuits have bounded fanin $O(1)$, which is an absolute constant.
\end{lemma}
\begin{proof} 

We apply Lemma \ref{lem:reduction} with $r = \lfloor \sqrt{r} \rfloor$, $s = \frac{n}{4r^2}$, $t = \frac{n}{r}$. Thus,
\begin{eqnarray*}
S_4\left(n, \frac{n}{4r^2}, \frac{n}{r}\right) & \le & S_2\left(\left\lfloor \frac{n}{\lfloor \sqrt{r} \rfloor} \right\rfloor, \frac{n}{4r^2}, \frac{n}{\lfloor \sqrt{r} \rfloor}\right) + O(n) \\
& \le & O\left(\frac{n}{\lfloor \sqrt{r} \rfloor} \log^{2}{(4r^2)}\right) + O(n) \quad\quad \text{By Lemma \ref{lem:ub_d2}} \\
& = & O(n) \ .
\end{eqnarray*}
The output gates have bounded fanin $O(1)$ by Lemma \ref{lem:reduction}.
\end{proof}

The following theorem is our main construction, which is proved by an induction on the depth.

\begin{theorem}
\label{thm:ub_main} Let $c_0$ be the constant in Lemma \ref{lem:condenser}. There exist absolute constants $c, D > 0$ such that the following holds.

For any $c_0 \le r \le n$, and for any $k \ge 3$, we have
\begin{equation}
\label{equ:d2k_first}
S_{2k}\left(n, \frac{n}{A(k-1, r)}, \frac{n}{r}\right) \;\le\; 2cn.
\end{equation}
Moreover, the output gates of the linear circuits encoding $\left(n, \frac{n}{A(k-1,r)}, \frac{n}{r}\right)$-PGC have bounded fanin $D$.

For any $2 \le r \le n$ and for any $k \ge 2$,
\begin{equation}
\label{equ:d2k_second}
S_{2k}\left(n, \frac{n}{r}, n\right) \;\le\; 3c\lambda_{2k}(r) \cdot n.
\end{equation}
(But the circuit encoding $\left(n, \frac{n}{r}, n\right)$-PGC does not necessarily have bounded output fanin.)
\end{theorem}

\begin{figure}[h]
\includegraphics[scale=0.6]{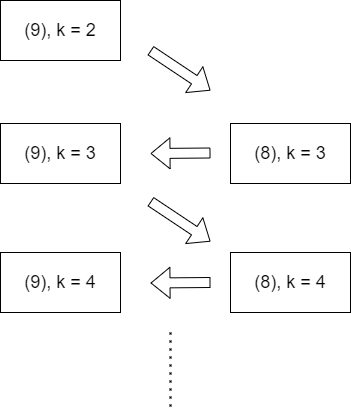}
\centering
\caption{Inductive proof of Theorem \ref{thm:ub_main}}
\label{fig:ub_induction_proof}
\end{figure}
\begin{proof} (of Theorem \ref{thm:ub_main}) The overall proof is an induction on $k$, as illustrated by Figure \ref{fig:ub_induction_proof}. 

We pick constants $D, c$ such that
\[
D \,\ge\, 2 c_2 \max(D_1, D_2)
\]
and
\[
c \,\ge\, \max\left(c_1 D + (\log c_0)^2 c_3, 2(2c_1D_1 + c_4 + c_5), c_6\right).
\]
Relevant constants used in the proof are summarized in Table \ref{table:constants} on page~\pageref{table:constants}.

The base step is \eqref{equ:d2k_second} for $k = 2$, which is true by Lemma \ref{lem:ub_gal}.

\begin{table}
\centering
\begin{tabular}{ |c|c| } 
 \hline
 $c_0$ & the constant in Lemma \ref{lem:condenser} and \ref{lem:reduction} \\ 
 \hline
 $c_1$ & \makecell{the constant in \eqref{equ:composition_depthd} in Lemma \ref{lem:composition}, i.e.,\\$S_{d}(n, r_1, r_{\ell+1}) \,\,\le\,\, \sum_{i = 1}^\ell s_i + c_2 D\ell n.$} \\ 
 \hline
 $c_2$ & \makecell{Lemma \ref{lem:composition}, the fanin of the output gates (of the depth-$d$ circuit\\ encoding the $(n, r_1, r_{\ell+1})$-PGC) is bounded by $c_2 \ell D$} \\ 
 \hline
 $c_3$ & Lemma \ref{lem:ub_d2}, $S_2\left(n, \frac{n}{r}, n\right) \le c_3 n \log^2 r$ \\ 
 \hline
 $c_4$ & \makecell{the constant in \eqref{equ:red_ineq} in Lemma \ref{lem:reduction}, i.e.,\\$S_d(n, s, t) \le S_{d-2}\left(\left\lfloor \frac{n}{r} \right\rfloor, s, \frac{n}{r}\right) + c_4 n$} \\ 
  \hline
 $c_5$ & \makecell{the constant in Lemma \ref{lem:handy}, i.e.,\\$S_4\left(n, \frac{n}{4r^2}, \frac{n}{r}\right) \le c_5 n$} \\ 
 \hline
  $c_6$ & \makecell{when $d = 4$, the constant in $O_d(\lambda_d(r) \cdot n)$ in Lemma \ref{lem:ub_gal}, i.e.,\\$S_4\left(n, \frac{n}{r}, n\right) \le c_6 \lambda_4(r) \cdot n$} \\ 
 \hline
   $D_1$ & Lemma \ref{lem:reduction}, output gates have fanin at most $D_1$ \\ 
 \hline
    $D_2$ & Lemma \ref{lem:handy}, output gates have fanin at most $D_2$ \\ 
 \hline
\end{tabular}
\caption{Absolute constants used in the proof.}
\label{table:constants}
\end{table}

For the induction step, we shall prove
\begin{itemize}
	\item \eqref{equ:d2k_first} implies \eqref{equ:d2k_second} for the same $k$, and
	\item \eqref{equ:d2k_second} with $k -1$ implies \eqref{equ:d2k_first} with $k$.
\end{itemize}

First, let us prove \eqref{equ:d2k_first} implies \eqref{equ:d2k_second} for the same $k$, where $k \ge 3$. By the composition lemma (i.e., Lemma \ref{lem:composition}),
\begin{eqnarray*}
& & S_{2k}\left(n, \frac{n}{r}, n\right) \\
& \le &  S_{2k}\left(n, \frac{n}{c_0}, n\right)+\sum_{i=2}^{h} S_{2k}\left(n, \frac{n}{A^{(i)}_{k-1}(c_0)}, \frac{n}{A^{(i-1)}_{k-1}(c_0)}\right) + c_1 Dhn \ ,
\end{eqnarray*}
where $h$ is the minimum integer such that $A^{(h)}_{k-1}(c_0) \ge r$ (and hence $A^{(h-1)}_{k-1}(c_0) < r$). Note that $A^{(h)}_{k-1}(c_0) \ge A^{(h)}_{k-1}(1) =A(k, h)$ and $A^{(h-1)}_{k-1}(c_0) \ge A(k, h-1)$. Since $A(k, h-1) \le A^{(h-1)}_{k-1}(c_0) < r$, by Proposition \ref{prop:ack_properties}, we have $h \le \lambda_{2k}(r)$. By Lemma \ref{lem:ub_d2}, $S_{2k}\left(n, \frac{n}{c_0}, n\right) \le c_3 n (\log c_0)^2$, and the circuit encoding $\left(n, \frac{n}{c_0}, n\right)$-PGC has output fanin 1, which can be done by adding an extra layer to the depth-2 circuit given by Lemma \ref{lem:ub_d2}. By induction hypothesis, $S_{2k}\left(n, \frac{n}{A^{(i)}_{k-1}(c_0)}, \frac{n}{A^{(i-1)}_{k-1}(c_0)}\right) \le 2cn$, and the fanin of the outputs in the circut encoding $\left(n, \frac{n}{A^{(i)}_{k-1}(c_0)}, \frac{n}{A^{(i-1)}_{k-1}(c_0)}\right)$-PGC is at most $D$.
Thus
\begin{eqnarray*}
S_{2k}\left(n, \frac{n}{r}, n \right) & \le & c_3 n (\log c_0)^2 + (\lambda_{2k}(r)-1) \cdot 2cn + c_1 D\lambda_{2k}(r)n \\
& \le &  2c \lambda_{2k}(r) n + \lambda_{2k}(r)n (c_1D + (\log c_2)^2 c_3) \\
& \le & 3c \lambda_{2k}(r) \cdot n,
\end{eqnarray*}
since $c \ge c_1D + (\log c_2)^2 c_3$.

\vspace{0.2cm}
Now, we prove \eqref{equ:d2k_second} with $k -1$ implies \eqref{equ:d2k_first} with $k$. That is, assuming $S_{2(k-1)}\left(n, \frac{n}{r}, n\right) \le 3c\lambda_{2(k-1)}(r) \cdot n$ for any $2 \le r \le n$, we will prove $S_{2k}\left(n, \frac{n}{A(k-1, r)}, \frac{n}{r}\right) \le 2cn$ for any $r\in [c_0, n]$, and the output gates of the linear circuit encoding $S_{2k}\left(n, \frac{n}{A(k-1, r)}, \frac{n}{r}\right)$-PGC have bounded fanin $\le D$.

Note that $r \ge c_0$ by our condition. Applying the reduction lemma (i.e., Lemma \ref{lem:reduction}) with $r' = 2r$, we have
\begin{eqnarray*}
& & S_{2k}\left(n,  \frac{n}{A(k-1, r)}, \frac{n}{4r^2}\right) \\
 & \le & S_{2(k-1)}\left(\left\lfloor \frac{n}{2r} \right\rfloor, \frac{n}{A(k-1, r)}, \frac{n}{2r}\right) + c_4 n \\& \le & 3c \cdot  \frac{n}{2r} \cdot \lambda_{2(k-1)}\left(\frac{A(k-1, r)}{2r}\right)  + c_4 n \quad\quad \text{By induction hypothesis}\\
& \le & \frac{3c}{2} \cdot n + c_4 n,
\end{eqnarray*}
where the last step $\lambda_{2(k-1)}\left(\frac{A(k-1, r)}{2r}\right) \le \lambda_{2(k-1)}(A(k-1, r)) = r$ is by Proposition \ref{prop:ack_properties}. By the reduction lemma (i.e., Lemma \ref{lem:reduction}), the output gates of the linear circuit encoding $\left(n,  \frac{n}{A(k-1, r)}, \frac{n}{4r^2}\right)$-PGC have bounded fanin $D_1$.

Applying the composition lemma (i.e., Lemma \ref{lem:composition}), we have
\begin{eqnarray*}
	S_{2k}\left(n,  \frac{n}{A(k-1, r)}, \frac{n}{r} \right) & \le & S_{2k}\left(n, \frac{n}{A(k-1, r)}, \frac{n}{4r^2}\right)  + S_{2k}\left(n, \frac{n}{4r^2}, \frac{n}{r}\right) + 2c_1D_1n \\
	& \le & \frac{3c}{2} \cdot n + c_4 n + c_5 n + 2c_1D_1n \\
 & \le & 2cn \ ,
\end{eqnarray*}
where the second last step is by Lemma \ref{lem:handy}, and the last step is because $c \ge 2(2c_1D_1 + c_4 + c_5)$. We claim the output gates of the overall linear circuit encoding $\left(n,  \frac{n}{A(k-1, r)}, \frac{n}{r} \right)$-PGC have fanin at most $2 c_2 \max(D_1, D_2)$, because
\begin{itemize}
	\item The output gates of the circuit encoding $\left(n,  \frac{n}{A(k-1, r)}, \frac{n}{4r^2}\right)$-PGC have bounded fanin $D_1$ by the reduction lemma (i.e., Lemma \ref{lem:reduction}). (Notice that the circuit encoding $\left(\left\lfloor \frac{n}{2r} \right\rfloor, \frac{n}{A(k-1, r)}, \frac{n}{2r}\right)$-PGC does not have bounded outputs fanin, which lies in the \emph{inner part} in the reduction lemma.)
	\item The output gates of the depth-$2k$ circuit encoding $\left(n, \frac{n}{4r^2}, \frac{n}{r}\right)$-PGC have bounded fanin $D_2$ by Lemma \ref{lem:handy}.
	\item In the composition lemma (i.e., Lemma \ref{lem:composition}), we compose 2 circuits, encoding $\left(n,  \frac{n}{A(k-1, r)}, \frac{n}{4r^2} \right)$-PGC and $\left(n, \frac{n}{4r^2}, \frac{n}{r}\right)$-PGC respectively. So the fanin of the outputs of the overall circuit is bounded by
 $
 2 c_2 \max(D_1, D_2),
 $
 as desired.
\end{itemize}
\end{proof}

By taking $r = n$ in Theorem \ref{thm:ub_main}, we immediately have

\begin{corollary}
\label{cor:linear_328}
For any $n$, 
\[
S_{\alpha(n)}(n) \;=\; O(n).
\]
\end{corollary}

We have constructed a linear circuit of size $O(n)$ and depth $\alpha(n)$ encoding a code $C:\{0, 1\}^n \to \{0, 1\}^{32n}$ with relative distance $\frac{1}{8}$. By putting a rate booster at the bottom, one can achieve any constant rate and relative distance within the Gilbert-Varshamov bound.

\begin{proof} (of Theorem \ref{thm:depth_lb}) By Corollary \ref{cor:linear_328}, there exists a linear circuit of size $O(\lambda_d(n) \cdot n)$ and depth $d$ encoding a code $C:\{0, 1\}^n \to \{0, 1\}^{32n}$ with relative distance $\frac{1}{8}$.

By Lemma \ref{lem:rate_booster} with $c = \frac{1}{r}$, $\gamma = \delta$, $\delta = 4$, we know there exists a probabilistic circuit $D : \{0, 1\}^{32n} \to \{0, 1\}^{\lfloor n/r \rfloor}$ of depth $1$ and size $O_{r, \delta}(n)$ such that
\[
\Pr_D\left[\wt(D(x)) \ge \delta \cdot \left\lfloor \frac{n}{r}\right\rfloor\right] \;>\; 1 - 2^{-n}
\]
for any $x \in \{0, 1\}^{32n}$ with $\wt(x) \ge 4n$. Applying a union bound over all $C(x)$, where $0 \not= x \in \{0, 1\}^n$, we claim there exists a (deterministic) linear circuit $D$ of depth $1$ and size $O_{r, \delta}(n)$ such that $\wt(D(C(x))) \ge \delta \cdot \lfloor \frac{n}{r}\rfloor$ for all nonzero $x \in \{0, 1\}^n$. In addition, the output gates in circuit $D$ have bounded fanin $O_{r, \delta}(1)$.

Note that the size of the circuit $D(C(x))$ is $O(n) + O_{r, \delta}(n) = O_{r, \delta}(n)$ and the depth of the circuit is $d + 1$. Notice that the output gates have bounded fanin $O_{r, \delta}(1)$. So collapsing the last layer (for the linear circuit) will increase the size by a multiplicative factor $O_{r, \delta}(1)$. As such, we obtain a circuit of depth $d$ and size $O_{r, \delta}(n)$.

In particular, by letting $d = \alpha(n)$, we obtain a circuit of depth $\alpha(n)$ and size $O_{r, \delta}(n)$.
\end{proof}

\section{Depth lower bound}
\begin{definition} \label{def:densely_regular} (Densely regular graph \cite{Pud94}) Let $G$ be a directed acyclic graph with $n$ inputs and $n$ outputs. Let $0 < \epsilon, \delta$ and $0 \le \mu \le 1$. We
say $G$ is \emph{$(\epsilon, \delta, \mu)$-densely regular} if for every $k \in [\mu n, n]$, there are probability distributions $\mathcal{X}$
and $\mathcal{Y}$ on $k$-element subsets of inputs and outputs respectively, such that for every $i \in [n]$,
\[
\Pr_{X \in \mathcal{X}}[i \in X] \;\le\; \frac{k}{\delta n} \ , \quad\quad \Pr_{Y \in \mathcal{Y}}[i \in Y] \;\le\; \frac{k}{\delta n} \ ,
\]
and the expected number of vertex-disjoint paths from $X$ to $Y$ is at least $\epsilon k$ for randomly chosen $X \in \mathcal{X}$ and $Y \in \mathcal{Y}$.

Denote by $D(n, d, \epsilon, \delta, \eta)$ the minimal size of a $(\epsilon, \delta, \mu)$-densely regular layered directed acyclic graph with $n$ inputs and $n$ outputs and depth $d$. 
\end{definition}

\begin{theorem}
\label{thm:pudlak_lb}
(Theorem 3 in \cite{Pud94}) Let $\epsilon, \delta > 0$. For every $d \ge 3$, and every $r \le n$,
\[
D\left(n, d, \epsilon, \delta, \frac{1}{r}\right) \;=\; \Omega_{d, \epsilon, \delta}\left(n \lambda_d(r)\right).
\]	
\end{theorem}

To apply the above lower bounds when $d$ is not \emph{fixed}, we need to figure out the hidden constant that depends on $d, \epsilon, \delta$.

\begin{theorem}
\label{thm:pudlak_lb_refinement}
Let $\epsilon, \delta > 0$. For every $d \ge 3$, and every $r \le n$,
\[
D\left(n, d, \epsilon, \delta, \frac{1}{r}\right) \;\ge\; \Omega\left(2^{-d/2} \epsilon \delta^2 \lambda_d(r) n\right).
\]	
\end{theorem}

The constant in $\Omega(2^{-d/2} \epsilon \delta^2 \lambda_d(r) n)$ is an absolute constant. The proof of Theorem \ref{thm:pudlak_lb_refinement} is almost the same as Theorem \ref{thm:pudlak_lb}, which is in Appendix \ref{app:pudlak_lb_ref}.

\begin{corollary}
\label{cor:code_dr}
(Corollary 15 in \cite{GHK+12}) Let $0 < \rho, \delta < 1$ be constants and $C$ be a circuit encoding an error-correcting code $\{0, 1\}^{\rho n} \to \{0, 1\}^n$ with relative distance at least $\delta$. If we extend the underlying graph with $(1-\rho)n$ dummy inputs, then its underlying graph is $(\rho \delta, \rho, \frac{1}{n})$-densely regular.
\end{corollary}

The proof of Theorem \ref{thm:depth_lb} readily follows from Theorem \ref{thm:pudlak_lb_refinement} and Corollary \ref{cor:code_dr}.

\begin{proof} (of Theorem \ref{thm:depth_lb}) By Corollary \ref{cor:code_dr}, the underlying graph $G$ of the circuit $C_n$, is $\left(\rho' \delta, \rho', \frac{1}{\lfloor n/\rho \rfloor}\right)$-densely regular with $\lfloor n/\rho \rfloor$ inputs and $\lfloor n/\rho \rfloor$ outputs, by adding dummy inputs, where
$
\rho' = \frac{n}{\lfloor n/\rho \rfloor} > \frac{1}{2} \rho.
$

Notice that graph $G$ is not necessarily layered. We can always make $G$ layered, by increasing the size by a factor of $d$, where $d$ is the depth of $G$. By slightly abusing notations, we use $G$ to denote a layered $(\rho' \delta, \rho', \frac{1}{\lfloor n/\rho \rfloor})$-densely regular graph of depth $d$, where $|E(G)| \le cdn$.

On the other hand, by Theorem \ref{thm:pudlak_lb_refinement}, we have
\[
|E(G)| \;\ge\; \Omega\left(2^{-d/2} \delta \rho'^3 \lambda_d(\lfloor n/\rho \rfloor) \lfloor n/\rho \rfloor\right).
\]
Combining with $|E(G)| \le cdn$, we have
\begin{equation}
\label{equ:lb_bound_lamd}
\lambda_d(\lfloor n/\rho \rfloor) \le O\left(\frac{1}{\delta \rho'^2} c d 2^{d/2}\right) = O_{\rho,\delta, c}\left(d 2^{d/2}\right).
\end{equation}

Therefore,
\begin{align*}
\lambda_{d+2}(n) \;\;\le\;\; &
	\lambda_{d+2}(\lfloor n/\rho \rfloor) \\
	  \;\;=\;\; & \lambda_d^*(\lfloor n/\rho \rfloor) \\
	 \;\;\le\;\; & \lambda_d^*(\lambda_d(\lfloor n/\rho \rfloor)) + 1 && \text{By Definition \ref{def:lambda}} \\
	 \;\;\le\;\; & \lambda_d^*(O_{\rho,\delta, c}(1) \cdot d 2^{d/2})  + 1 && \text{By \eqref{equ:lb_bound_lamd}}\\
	 \;\;\le\;\; & \lambda_d^*(\lambda_d(O_{\rho,\delta, c}(1) \cdot d 2^{d/2})) + 2 && \text{By Definition \ref{def:lambda}} \\
	 \;\;\le\;\; & \lambda_d^*(d) + 2 && \text{By Proposition \ref{prop:lambda_dd}}\\
	 \;\;\le\;\; & 6. && \text{By Proposition \ref{prop:inv_ack_def_properties}}
\end{align*}
So, we conclude $d+2 \ge \alpha(n)$, i.e., $d \ge \alpha(n) - 2$.
\end{proof}

We would like to point out that, alternatively, one can use a powerful lemma by Raz and Shipilka \cite{RS03}, to prove the depth lower bound  (i.e., Theorem \ref{thm:depth_lb}).

\section{Superconcentrator-induced codes}

It is already known that circuits for encoding error-correcting codes must satisfy some superconcentrator-like connectivity properties. For example, Speilman \cite{Spi96} observed that any circuits encoding codes from $\{0, 1\}^n$ to $\{0, 1\}^m$ with distance $\delta m$ must have $\delta n$ vertex-disjoint paths connecting any chosen $\delta n$ inputs to any set of $(1-\delta)m$ outputs; using matroid theory, G\'{a}l \emph{et al.} \cite{GHK+12} proved that, for any $k \le n$, for any $k$-element subset of inputs $X$, taking a random $k$-element subset of outputs $Y$, the expected number of vertex-disjoint paths from $X$ to $Y$ is at least $\delta k$.

We observe a connection in the \emph{reverse} direction by showing that \emph{any} superconcentrator graph, converted to an arithmetic circuit over a sufficiently large field, can encode a good code. 
(Recall that a directed acyclic graph $G = (V, E)$ with $m$ inputs and $n$ outputs is an \emph{$(m, n)$-superconcentrator}, if for any equal-size inputs $X \subseteq [m]$ and outputs $Y \subseteq [n]$, the number of vertex-disjoint paths from $X$ to $Y$ is $|X|$.) Furthermore, the code $C: \mathbb{F}^n \to \mathbb{F}^m$ (encoded by the above circuits) satisfies a distance criterion, stronger than MDS (maximum distance separable) codes, captured by the following definition.


\begin{definition} 
\label{def:sc_codes}
(Superconcentrator-induced code) $C : \mathbb{F}^n \to \mathbb{F}^m$ is a \emph{\superconcentrator code} if
\[
\dist(C(x), C(y)) \;\ge\; m - \dist(x, y) + 1
\]
for any distinct $x, y \in \mathbb{F}^n$.
\end{definition}

For a linear code $C : \mathbb{F}^n \to \mathbb{F}^m$, it is well known that $C$ is an MDS code, i.e., $C$ satisfies the Singleton bound
\[
\dist(C(x), C(y)) \;\ge\; m - n + 1 \quad \forall x \not= y \in \mathbb{F}^n,
\]
if and only if any $n$ columns of its generator matrix are linearly independent. Similarly, we show that a linear code is a \superconcentrator code if and only if every square submatrix of its generator matrix is nonsingular. (Such matrices are sometimes called \emph{totally invertible matrices}.)

\begin{lemma}
\label{lem:linear_algebra_cha}
Linear code $C : \mathbb{F}^n \to \mathbb{F}^m$ is a \superconcentrator code if and only if every square submatrix of its generator matrix is nonsingular.
\end{lemma}
\begin{proof}
For the ``if'' direction, assume for contradiction that there exists a nonzero $x \in \mathbb{F}^n$ such that $\wt(C(x)) \le m - \wt(x)$. Let $X = \supp(x) \subseteq [n]$, and let $Y$ be any subset of $[m] \setminus \supp(C(x))$ of size $|X|$. Denote the generator matrix by $M \in \mathbb{F}^{n \times m}$. By the definition of the generator matrix, we have $C(x) = x M$. Thus,
\begin{align*}
C(x)|_Y \;=\; & x M_{[n], Y} \\
\;=\; & x|_X M_{X, Y} + x|_{[n] \setminus X} M_{[n] \setminus X, Y} \\
\;=\; & x|_X M_{X, Y} && {\text{since } x|_{[n] \setminus X} = 0,}
\end{align*}
where $M_{X, Y}$ denotes the $|X| \times |Y|$ submatrix indexed by rows $X$ and columns $Y$. Since matrix $M_{X, Y}$ is nonsingular, and $x|_X$ is nonzero, we claim $C(x)|_Y$ is nonzero. This contradicts the definition of $Y$.

For the ``only if'' direction, we prove the contraposition. That is, if there exists a singular square submatrix, then $C$ is not a \superconcentrator code. Suppose there exist equal-size $X \subseteq [n]$ and $Y \subseteq [m]$ such that $M_{X, Y}$ is singular. Since $\rank M_{X, Y} < |X|$, there exists a nonzero $x' \in \mathbb{F}^{|X|}$ such that $x' M_{X, Y} = 0$. Extend $x'$ to a vector $x \in \mathbb{F}^n$ by adding zeros on the coordinates outside of $X$, that is, $x' = x|_{\supp(x)}$. As such, we have
\begin{align*}
	C(x)|_Y \;=\; & (x M)|_Y \\
	 \;=\; & x M_{[n], Y} \\
	 \;=\; & x|_{X} M_{X, Y} + x|_{[n] \setminus X} M_{[n] \setminus X, Y} \\
	 \;=\; & x|_{X} M_{X, Y} && \text{ since $x_{[n] \setminus X}$ is zero} \\
	 \;=\; & x' M_{X, Y} \\
	 \;=\; & 0,
\end{align*}
which implies that $\wt(C(x)) \le m - |Y| = m - \wt(x)$. That is, $C$ is not a \superconcentrator code.
\end{proof}

We justify the naming ``\superconcentrator codes'' by proving
\begin{itemize}
	\item Any (unrestricted) arithmetic circuit encoding a \superconcentrator code, viewed as a graph, must be a superconcentrator.
	\item Any superconcentrator can be converted to an arithmetic circuit encoding a \superconcentrator code, by replacing each vertex with an addition gate over a sufficiently large finite field, and choosing the coefficients uniformly at random. (The field size is exponential in the size of the graph.)
\end{itemize}

We consider the following unrestricted arithmetic circuit model encoding a \superconcentrator code $C:\mathbb{F}^n \to \mathbb{F}^m$:
\begin{itemize}
	\item Let $\mathbb{F}$ be a finite field.
	\item There are $n$ inputs, denoted by $x_1, x_2, \ldots, x_n \in \mathbb{F}$; there are $m$ outputs, denoted by $y_1, y_2, \ldots, y_m \in \mathbb{F}$.
	\item Gates have unbounded fanin and fanout. The size of the circuit is defined to be the number of wires, instead of gates.
	\item Each internal and output gate can compute \emph{any} function from $\mathbb{F}^s \to \mathbb{F}$, where $s$ is the fanin of the gate (unbounded). 
\end{itemize}
In this model, any single output function from $\mathbb{F}^n$ to $\mathbb{F}$ can be computed by a single gate. It makes sense to consider the circuit complexity of multiple output functions.

\begin{lemma} Any unrestricted arithmetic circuit encoding a \superconcentrator code must be a superconcentrator.
\end{lemma}
\begin{proof} Let $C:\mathbb{F}^n \to \mathbb{F}^m$ be an unrestricted arithmetic circuit encoding a \superconcentrator code. Let $G$ denote the underlying directed acyclic graph with $n$ inputs and $m$ outputs.

Suppose for contradiction that $G$ is not an $(n, m)$-superconcentrator, that is, there exist equal-size inputs $S \subseteq [n]$ and outputs $T \subseteq [m]$ such that the number of vertex-disjoint paths from $S$ to $T$ is less than $|S|$. By Menger's theorem, there exists a cut vertex set, denoted by $V$, whose removal disconnects $S$ and $T$ and $|V| < |S|$.

Fix the inputs outside of $S$ arbitrarily, while leave the inputs in $S$ undetermined. So, the outputs in $T$ are completely determined by the gates in $V$, the ``cut vertices'' between $S$ and $T$. Since $|V| < |S|$, the outputs in $T$ can take at most $|\mathbb{F}|^{|V|} < |\mathbb{F}|^{|S|}$ values. By the pigeonhole principle, there exist $x, x' \in \mathbb{F}^n$ such that 
\begin{itemize}
	\item $x|_{S} \not= x'|_{S}$,
	\item $x|_{[n] \setminus S} = x'|_{[n] \setminus S}$, and
	\item $C(x)|_T = C(x')|_T$.
\end{itemize}
Therefore, $\dist(C(x), C(x')) \le m - |T| = m - |S|$, and $\dist(x, x') \le |S|$, which implies $\dist(C(x), C(x')) \le m -  \dist(x, x') $. This contradicts the definition of \superconcentrator codes.
\end{proof}


In the reverse direction, we prove the following.

\begin{lemma}
\label{lem:sc_codes_on_sc}
Let $G$ by an $(n, m)$-superconcentrator. Let $C_G : \mathbb{F}^n \to \mathbb{F}^m$ be an arithmetic circuit by replacing each vertex in $G$ with an addition gate, and choosing the coefficient on each edge uniformly and random (over the finite field $\mathbb{F}$). With probability at least
$
1-\sum_{i = 1}^n {n \choose i}{m \choose i} \frac{di}{|\mathbb{F}|},
$ $C_G$ encodes a \superconcentrator code.
\end{lemma}
\begin{proof}
Denote the inputs by $x_1, x_2, \ldots, x_n$, and outputs by $y_1, y_2, \ldots, y_m$. For each edge $e \in E(G)$, assign a random coefficient $r_e$. Let $M$ be the $n \times m$ generator matrix of the linear code $C_G$. One can see that
\[
M_{i, j} \;=\; \sum_{\text{path $p$ from $x_i$ to $y_j$}} \prod_{e \in E(p)} r_e,
\]	
where $p$ enumerates \emph{all} paths from $x_i$ to $y_j$. View $M_{i, j}$ as a (multilinear) polynomial in $\mathbb{F}[\{r_e\}_{e \in E(G)}]$. Note that the degree of $M_{i, j}$ is exactly the length of the longest path from $x_i$ to $y_j$.

\begin{claim}
\label{claim:det0}
For any equal-size subsets $X \subseteq [n]$ and $Y \subseteq [m]$,
\[
\Pr[\det(M_{X, Y}) = 0] \;\le\; \frac{d|X|}{|\mathbb{F}|},
\]	
where $d$ is the depth of the superconcentrator.
\end{claim}
\begin{proof} (of Claim \ref{claim:det0}) Note that $\det(M_{X, Y})$ is a polynomial of degree at most $d|X|$, since each entry is of degree at most $d$. By the definition of superconcentrator, there are $|X|$ vertex-disjoint paths connecting $X$ and $Y$. Setting $r_e$ to 1 for all edges $e$ on these $|X|$ vertex-disjoint paths, and leaving all others 0, the polynomial $\det(M_{X, Y})$ evaluates to $\pm 1$, which implies that $\det(M_{X, Y})$ is a nonzero polynomial. The claim follows from the Schwarz-Zipple lemma.
\end{proof}

By Lemma \ref{lem:linear_algebra_cha}, $C_G$ encodes a \superconcentrator code if and only if all square submatrix of $M$ is nonsingular. Taking a union bound over all equal-size $X \subseteq [n]$ and $Y \subseteq [m]$, we have
\[
\Pr[C_G \text{ encodes a \superconcentrator code}] \;\ge\; 1 - \sum_{i=1}^n {n \choose i}{m \choose i} \frac{di}{|\mathbb{F}|}.
\]
\end{proof}

\section{Conclusion}

In this work, we determine, up to an additive constant 2, the minimum depth required to encode asymptotically good error-correcting codes, i.e., codes with constant rate and constant relative distance, using a linear number of wires. The minimum depth is between $\alpha(n)-2$ and $\alpha(n)$, where $\alpha(n)$ is a version of the inverse Ackermann function. The upper bound is met by certain binary codes we construct (building on G\'{a}l \emph{et al.} \cite{GHK+12} with a few new ingredients) for any constant rate and constant relative distance within the Gilbert-Varshamov bounds. The lower bound applies to any constant rate and constant relative distance. We credit the lower bound to G\'{a}l \emph{et al.} \cite{GHK+12}, although not explicitly claimed or fully verified in \cite{GHK+12}; because our contribution is a routine checking of details.

Valiant articulated graph-theoretic arguments for proving circuit lower bounds \cite{Val75, Val76, Val77}. Since then, there have been fruitful results along this line. We show a result in the reverse direction, that is, we prove that any superconcentrator (after being converted to an arithmetic circuit) can encode a good code over a sufficiently large field (exponential in the size of the superconcentrator graph).

\printbibliography

@article{Bop97,
  author    = {Ravi B. Boppana},
  title     = {The average sensitivity of bounded-depth circuits},
  journal   = {Information Processing Letters},
  volume    = {63},
  number    = {5},
  pages     = {257-261},
  year      = {1997}
}

@inproceedings{LV11,
  title         = {Bounded-depth circuits cannot sample good codes},
  author        = {Shachar Lovett and
                   Viola, Emanuele},
  booktitle     = {Computational Complexity (CCC), 2011 IEEE 26th Annual Conference on},
  pages         = {243--251},
  year          = {2011},
  organization  = {IEEE}
}

@article{BM05,
  title         = {Endcoding complexity versus minimum distance},
  author        = {Louay M. J. Bazzi and
                   Sanjoy K. Mitter},
  journal       = {IEEE Transactions on Information Theory},
  volume        = {51},
  number        = {6},
  pages         = {2103--2112},
  year          = {2005},
  publisher     = {IEEE}
}

@inproceedings{BIL12,
  title         = {Large deviation bounds for decision trees and sampling lower bounds for AC0 circuits},
  author        = {Chris Beck and
                   Russell Impagliazzo and
                   Shachar Lovett},
  booktitle     = {Foundations of Computer Science (FOCS), 2012 IEEE 53rd Annual Symposium on},
  pages         = {101--110},
  year          = {2012},
  organization  = {IEEE}
}

@inproceedings{DGP73,
  title         = {On complexity of coding},
  author        = {Roland L. Dobrushin and
                   S. I. Gelfand and
                   Mark S. Pinsker},
  booktitle     = {Proc. 2nd Internat. Symp. on Information Theory},
  pages         = {174--184},
  year          = {1973}
}

@article{Spi96,
  title         = {Linear-time encodable and decodable error-correcting codes},
  author        = {Spielman, Daniel A},
  journal       = {IEEE Transactions on Information Theory},
  volume        = {42},
  number        = {6},
  pages         = {1723--1731},
  year          = {1996},
  publisher     = {IEEE}
}

@article{SS96,
  title         = {Expander codes},
  author        = {Sipser, Michael and Spielman, Daniel A},
  journal       = {IEEE transactions on Information Theory},
  volume        = {42},
  number        = {6},
  pages         = {1710--1722},
  year          = {1996},
  publisher     = {IEEE}
}

@inproceedings{DDP+83,
  author    = {Danny Dolev and
               Cynthia Dwork and
               Nicholas Pippenger and
               Avi Wigderson},
  title     = {Superconcentrators, generalizers and generalized connectors with limited
               depth},
  booktitle = {Proceedings of the 15th Annual {ACM} Symposium on Theory of Computing,
               25-27 April, 1983, Boston, Massachusetts, {USA}},
  pages     = {42--51},
  year      = {1983},
  crossref  = {DBLP:conf/stoc/STOC15},
%  url       = {http://doi.acm.org/10.1145/800061.808731},
%  doi       = {10.1145/800061.808731},
  timestamp = {Mon, 17 Oct 2011 17:25:06 +0200},
  biburl    = {http://dblp.org/rec/bib/conf/stoc/DolevDPW83},
  bibsource = {dblp computer science bibliography, http://dblp.org}
}

@article{RT00,
  title     = {Bounds for dispersers, extractors, and depth-two superconcentrators},
  author    = {Radhakrishnan, Jaikumar and Ta-Shma, Amnon},
  journal   = {SIAM Journal on Discrete Mathematics},
  volume    = {13},
  number    = {1},
  pages     = {2--24},
  year      = {2000},
  publisher = {SIAM}
}

@inproceedings{CT19,
  title     = {Bootstrapping results for threshold circuits “just beyond” known lower bounds},
  author    = {Chen, Lijie and Tell, Roei},
  booktitle = {Proceedings of the 51st Annual ACM SIGACT Symposium on Theory of Computing},
  pages     = {34--41},
  year      = {2019}
}

@article{AP94,
  author    = {Noga Alon and
               Pavel Pudl{\'{a}}k},
  title     = {Superconcentrators of depths 2 and 3; odd levels help (rarely)},
  journal   = {J. Comput. Syst. Sci.},
  volume    = {48},
  number    = {1},
  pages     = {194--202},
  year      = {1994},
%  url       = {https://doi.org/10.1016/S0022-0000(05)80027-3},
%  doi       = {10.1016/S0022-0000(05)80027-3},
  timestamp = {Sat, 20 May 2017 00:25:53 +0200},
  biburl    = {http://dblp.org/rec/bib/journals/jcss/AlonP94},
  bibsource = {dblp computer science bibliography, http://dblp.org}
}

@article{GHK+12,
  title     = {Tight Bounds on Computing Error-Correcting Codes by Bounded-Depth Circuits With Arbitrary Gates},
  author    = {Gal, Anna and Hansen, Kristoffer Arnsfelt and Koucky, Michal and Pudlak, Pavel and Viola, Emanuele},
  journal   = {IEEE Transactions on Information Theory},
  volume    = {59},
  number    = {10},
  pages     = {6611--6627},
  year      = {2013},
  publisher = {IEEE Press}
}

@article{Pud94,
  title     = {Communication in bounded depth circuits},
  author    = {Pudlak, Pavel},
  journal   = {Combinatorica},
  volume    = {14},
  number    = {2},
  pages     = {203--216},
  year      = {1994},
  publisher = {Springer}
}

@inproceedings{Val77,
  author    = {Leslie G. Valiant},
  title     = {Graph-theoretic arguments in low-level complexity},
  booktitle = {Mathematical Foundations of Computer Science 1977, 6th Symposium,
               Tatranska Lomnica, Czechoslovakia, September 5-9, 1977, Proceedings},
  pages     = {162--176},
  year      = {1977}
}

@article{GLS+21,
  title     = {Brakedown: Linear-time and post-quantum SNARKs for R1CS},
  author    = {Golovnev, Alexander and Lee, Jonathan and Setty, Srinath and Thaler, Justin and Wahby, Riad S},
  journal   = {Cryptology ePrint Archive},
  year      = {2021}
}

@book{Rom92,
  title     = {Coding and information theory},
  author    = {Roman, Steven},
  volume    = {134},
  year      = {1992},
  publisher = {Springer Science \& Business Media}
}

@article{LYC03,
  title     = {Linear network coding},
  author    = {Li, S-YR and Yeung, Raymond W and Cai, Ning},
  journal   = {IEEE transactions on information theory},
  volume    = {49},
  number    = {2},
  pages     = {371--381},
  year      = {2003},
  publisher = {IEEE}
}

@article{ACL+00,
  title     = {Network information flow},
  author    = {Ahlswede, Rudolf and Cai, Ning and Li, S-YR and Yeung, Raymond W},
  journal   = {IEEE Transactions on information theory},
  volume    = {46},
  number    = {4},
  pages     = {1204--1216},
  year      = {2000},
  publisher = {IEEE}
}

@article{Vio19,
  title={Lower bounds for data structures with space close to maximum imply circuit lower bounds},
  author={Viola, Emanuele},
  journal={Theory of Computing},
  volume={15},
  number={1},
  pages={1--9},
  year={2019},
  publisher={Theory of Computing Exchange}
}

@article{Li23,
  title={Secret Sharing on Superconcentrator},
  author={Li, Yuan},
  journal={arXiv preprint arXiv:2302.04482},
  year={2023}
}

@inproceedings{Val75,
  title={On non-linear lower bounds in computational complexity},
  author={Valiant, Leslie G},
  booktitle={Proceedings of the seventh annual ACM symposium on Theory of computing},
  pages={45--53},
  year={1975}
}

@article{CCL13,
  title={Fast matrix rank algorithms and applications},
  author={Cheung, Ho Yee and Kwok, Tsz Chiu and Lau, Lap Chi},
  journal={Journal of the ACM (JACM)},
  volume={60},
  number={5},
  pages={1--25},
  year={2013},
  publisher={ACM New York, NY, USA}
}

@article{Val76,
  title={Graph-theoretic properties in computational complexity},
  author={Valiant, Leslie G},
  journal={Journal of Computer and System Sciences},
  volume={13},
  number={3},
  pages={278--285},
  year={1976},
  publisher={Academic Press}
}

@book{LC04,
  title={Coding theory: a first course},
  author={Ling, San and Xing, Chaoping},
  year={2004},
  publisher={Cambridge University Press}
}

@article{GRS12,
  title={Essential coding theory},
  author={Guruswami, Venkatesan and Rudra, Atri and Sudan, Madhu},
  journal={Draft available at http://www. cse. buffalo. edu/atri/courses/coding-theory/book},
  volume={2},
  number={1},
  year={2012}
}

@proceedings{DBLP:conf/stoc/STOC15,
  editor    = {David S. Johnson and
               Ronald Fagin and
               Michael L. Fredman and
               David Harel and
               Richard M. Karp and
               Nancy A. Lynch and
               Christos H. Papadimitriou and
               Ronald L. Rivest and
               Walter L. Ruzzo and
               Joel I. Seiferas},
  title     = {Proceedings of the 15th Annual {ACM} Symposium on Theory of Computing,
               25-27 April, 1983, Boston, Massachusetts, {USA}},
  publisher = {{ACM}},
  year      = {1983},
  bibsource = {dblp computer science bibliography, http://dblp.org}
}

@article{RS03,
  author    = {Ran Raz and
               Amir Shpilka},
  title     = {Lower Bounds for Matrix Product in Bounded Depth Circuits with Arbitrary
               Gates},
  journal   = {{SIAM} J. Comput.},
  volume    = {32},
  number    = {2},
  pages     = {488--513},
  year      = {2003},
  url       = {https://doi.org/10.1137/S009753970138462X},
  doi       = {10.1137/S009753970138462X},
  bibsource = {dblp computer science bibliography, http://dblp.org}
}

@inproceedings{Dru14,
  title={Linear-time encodable codes meeting the gilbert-varshamov bound and their cryptographic applications},
  author={Druk, Erez and Ishai, Yuval},
  booktitle={Proceedings of the 5th conference on Innovations in theoretical computer science},
  pages={169--182},
  year={2014}
}

@inproceedings{Ish08,
  title={Cryptography with constant computational overhead},
  author={Ishai, Yuval and Kushilevitz, Eyal and Ostrovsky, Rafail and Sahai, Amit},
  booktitle={Proceedings of the fortieth annual ACM symposium on Theory of computing},
  pages={433--442},
  year={2008}
}

@inproceedings{CW89,
  title={Dispersers, deterministic amplification, and weak random sources},
  author={Cohen, Aviad and Wigderson, Avi},
  booktitle={30th Annual Symposium on Foundations of Computer Science},
  pages={14--19},
  year={1989},
  organization={IEEE Computer Society}
}

@article{Sip88,
  title={Expanders, randomness, or time versus space},
  author={Sipser, Michael},
  journal={Journal of Computer and System Sciences},
  volume={36},
  number={3},
  pages={379--383},
  year={1988},
  publisher={Elsevier}
}

@article{Tar75,
  title={Efficiency of a good but not linear set union algorithm},
  author={Tarjan, Robert Endre},
  journal={Journal of the ACM (JACM)},
  volume={22},
  number={2},
  pages={215--225},
  year={1975},
  publisher={ACM New York, NY, USA}
}

@book{Li17,
  title={Some Results in Low-Depth Circuit Complexity},
  author={Li, Yuan},
  year={2017},
  publisher={The University of Chicago}
}

@inproceedings{Lov18,
  title={MDS matrices over small fields: A proof of the GM-MDS conjecture},
  author={Lovett, Shachar},
  booktitle={2018 IEEE 59th Annual Symposium on Foundations of Computer Science (FOCS)},
  pages={194--199},
  year={2018},
  organization={IEEE}
}

\appendix

\section{Properties of the inverse Ackermann-type functions}
\label{app:inv_Ack_properites}

\begin{proposition}
\label{prop:lambda_rs}
(Claim 2.4 in \cite{RS03}) 
	\item[(i).] Each $\lambda_i(n)$ is a monotone function tending to infinity with $n$.
	\item[(ii).] For $i \ge 2$, $\lambda_{2i+1}(n) \le \lambda_{2i}(n) \le 2\lambda_{2i+1}(n)$.
	\item[(iii).] For $i \ge 2$ and $n \ge 2^7$, $\lambda_i(n) \le \lfloor \sqrt{\frac{n}{2}}\rfloor$.
\end{proposition}

The following proposition is given as Proposition 2 and 3 in the subsequent paper \cite{Li23}, but was developed first in the present work.

\begin{proposition}
\label{prop:inv_ack_def_properties}
\begin{itemize}
	\item[(i)] For any $d \ge 1$, and for all $n \ge 4$,
\[
\lambda_d(n) \;\le\; n-2.
\]
\item[(ii)] For any $d \ge 1$,
\[
\lambda_d(d) \;\le\; 4.
\]	
\end{itemize}
\end{proposition}
\begin{proof}
(i) We prove by induction on $d$. When $d = 1$ or $2$, it is easy to verify that $\lambda_d(n) \le n - 2$. Assuming $\lambda_d(n) \le n-2$ for all $n \ge 4$, let us prove that $\lambda_{d+2}(n) \le n-2$. By the definition of $\lambda_{d+2}(n)$, we have
\[
\lambda_{d+2}(n) \;=\; \lambda_{d}^*(n) \;\le\; \lambda_{d}(n) \;\le\; n - 2.
\]

(ii) We prove by induction on $d$. When $d \le 3$, it is easy to verify by straightforward calculation. Assuming $\lambda_d(d) \le 4$ is true for $d$, where $d \ge 2$, let us prove $\lambda_{d+2}(d+2) \le 4$. We have
\begin{align*}
\lambda_{d+2}(d+2) \;\;=\;\; & \lambda_{d}^*(d+2) \\
\;\;=\;\; & \lambda_{d}^*\left(\lambda_d(d+2)\right) + 1 \\
\;\;\le\;\; & \lambda_{d}^*(d) + 1 && \text{Proposition \ref{prop:inv_ack_def_properties} (i)} \\ 
\;\;=\;\; & \lambda_{d}^*(\lambda_d(d)) + 2 \\ 
\;\;\le\;\; & \lambda_{d}^*(4) + 2  && \text{By induction hypothesis} \\ 
\;\;\le\;\; & \max\{\lambda^*_2(4), \lambda^*_3(4)\} + 2 \\
\;\;=\;\; & 4.
\end{align*}
\end{proof}

\begin{proposition}
\label{prop:ack_properties}
For any $i, d \ge 1$,
\begin{itemize}
    \item[(1).] $A(i+1, d) = A^{(d)}_i(1)$.
	\item[(2).] $\lambda_{2i}(A(i, d)) = d$.
	\item[(3).] $\lambda_{2i}(A(i, d) + 1) = d + 1$.
	\item[(4).] $A(i, d) = \max\{n : \lambda_{2i}(n) \le d \}.$
	\item[(5).] $A(i, \lambda_{2i}(d)) \ge d$.
\end{itemize}
\end{proposition}
\begin{proof}
(1). By Definition \ref{def:acker},
\begin{eqnarray*}
A(i+1, d) & = & A(i, A(i+1, d-1)) \\
& = & A(i, A(i, A(i+1, d-2))) \\
& = & \underbrace{A(i, A(i, \ldots}_{d-1 \text{ times}}, A(i+1, 1) \ldots )) \\
& = & A_i^{(d)}(1),
\end{eqnarray*}
where we write $A(i, j)$ as $A_i(j)$.

(2). Prove by induction on $i$. For the base case $i = 1$, we have $A(1, d) = 2^d$, and $\lambda_2(n) = \lceil \log_2 n \rceil$. It is obvious that the $\lambda_2(A(1, d)) = d$ holds.

For the induction step, assuming the conclusion holds for $i$, let us prove it for $i+1$.
\begin{align*}
& \lambda_{2(i+1)}(A(i+1, d)) \\
\;\;=\;\; & \lambda_{2i}^*\left(A^{(d)}_i(1)\right) \\
\;\;=\;\; & \lambda_{2i}^*\left(A^{(d-1)}_i(1)\right) + 1 && \text{By induction hypothesis} \\ 
\;\;=\;\; & \lambda_{2i}^*\left(A^{(d-2)}_i(1)\right) + 2 && \text{By induction hypothesis} \\ 
\vdots & \\
\;\;=\;\; & d.
\end{align*}

(3). Prove by induction on $i$. For the base case $i = 1$, we have $A(1, d) = 2^d$, and $\lambda_2(n) = \lceil \log_2 n \rceil$. It is obvious that the $\lambda_2(A(1, d)+1) = d+1$ holds.

For the induction step, assuming the conclusion holds for $i$, let us prove it for $i+1$.
\begin{align*}
& \lambda_{2(i+1)}(A(i+1, d)+1) \\
\;\;=\;\; & \lambda_{2i}^*\left(A^{(d)}_i(1)+1\right) \\
\;\;=\;\; & \lambda_{2i}^*\left(A\left(i, A^{(d-1)}_i(1)\right)+1\right) \\
\;\;=\;\; & \lambda_{2i}^*\left(A^{(d-1)}_i(1)+1\right) + 1 && \text{By induction hypothesis} \\ 
\;\;=\;\; & \lambda_{2i}^*\left(A^{(d-2)}_i(1)+1\right) + 2 && \text{By induction hypothesis} \\ 
\vdots & \\
\;\;=\;\; & \lambda_{2i}^*(1 + 1) + d \\
\;\;=\;\; & d + 1.
\end{align*}

(4) follows from (2) and (3).

(5) follows from (4).
\end{proof}

\begin{proposition}
\label{prop:lambda_dd}
Let $c > 0$ be a constant. When $d$ is large enough, $\lambda_d(c \cdot d 2^{d/2}) \le d$.
\end{proposition}
\begin{proof}
If $d$ is even, i.e., $d = 2i$, by Proposition \ref{prop:ack_properties}, it suffices to prove
\begin{equation}
\label{equ:goal_lb}
c (2i) 2^i \;\le\; A_i(2i).	
\end{equation}
Notice that $A_1(2i) = 2^{2i}$ and $A_{i+1}(n) \ge A_i(n)$. Clearly, \eqref{equ:goal_lb} is true when $i$ is large enough, and this sufficiently large number only depends on $c$.

If $d$ is odd, i.e., $d = 2i+1$, it suffices to prove $\lambda_{2i}(c \cdot d 2^{d/2}) \le d$, as $\lambda_{2i+1}(n) \le \lambda_{2i}(n)$ by Proposition \ref{prop:lambda_rs}. By Proposition \ref{prop:ack_properties}, it suffices to prove
\[
c (2i+1) 2^{i+\frac{1}{2}} \;\le\; A_i(2i+1),
\]
which is true because $A_1(2i+1) = 2^{2i+1}$.
\end{proof}

\section{Proof of Theorem \ref{thm:pudlak_lb_refinement}}
\label{app:pudlak_lb_ref}

\begin{lemma}
\label{lem:pudlak_d1}
(Lemma 8 in \cite{Pud94}) For any $1 \le r \le n$, and for any $\epsilon, \delta \in (0, 1]$, we have
\[
D\left(n, 1, \epsilon, \delta, \frac{1}{r}\right) \;\ge\; \epsilon \delta^2 nr.
\]
\end{lemma}

\begin{lemma}
\label{lem:pudlak_f_star}	
 (Lemma 5 in \cite{Pud94}) Suppose $f(0) = f(1) = 0$ and $f(n) \le \lfloor \sqrt{n} \rfloor$ for every $n > 1$. For every $n \ge 0$, we have
\begin{itemize}
\item[(i)]	$f^*(n) \le f(n) \le  \lfloor \sqrt{n} \rfloor,$
\item[(ii)] $\frac{f^{(i)}(n)}{f^{(i+1)}(n)} \ge f^{(i+1)}(n)$
for every $i \ge 1$ provided the denominator is not 0,
\item[(iii)] $f^{(i)}(n) \ge \frac{f^*(n)}{2}$ for every $i \le \frac{f^*(n)}{2}$.
\end{itemize}

Indeed, (iii) can be strengthened. 

\begin{lemma}
\label{lem:f_star} Suppose $f(n) \le \lfloor \sqrt{n} \rfloor$. For every $n \ge 1$ and every $i \le \frac{f^*(n)}{2}$, we have $f^{(i)}(n) \ge f^*(n)$.
\end{lemma}
\begin{proof} Let us assume $f^*(n) \ge 2$. Otherwise, $f^*(n) \le 1$, that is, $f(n) \le 1$. Note that $i \le f^*(n)/2 \le 1/2$, and thus $i = 0$. So $f^{(0)}(n) = n \ge f^*(n) = 1$ always holds.

By the definition of $f^*(n)$ and $f(n) \le \lfloor \sqrt{n} \rfloor$, we have
\begin{eqnarray*}
f^{(f^*(n))}(n) & = & 1, \\
f^{(f^*(n)-1)}(n) & \ge & 2, \\
f^{(f^*(n)-2)}(n) & \ge & 2^2, \\
f^{(f^*(n)-3)}(n) & \ge & 2^4, \\
\ldots \\
f^{(f^*(n)-j)}(n) & \ge & 2^{2^{j-1}}
\end{eqnarray*}	
for any $1 \le j \le f^*(n)$.

When $i \le \frac{f^*(n)}{2}$, $f^*(n) - i \ge \frac{f^*(n)}{2}$.  Thus,
\begin{eqnarray*}
f^{(i)}(n) & \ge & 2^{2^{f^*(n) - i - 1}} \\
& \ge & 2^{2^{\frac{f^*(n)}{2} - 1}} \\
& \ge & f^*(n),
\end{eqnarray*}
where the last step is because $2^{2^{\frac{x}{2}-1}} \ge x$ for all real $x \ge 4$. (If $f^*(n) = 2$, then $f^{(1)}(n) \ge 2$. If $f^*(n) = 3$, then $f^{(1)}(n) \ge 4$.)
\end{proof}

\end{lemma}

The following lemma is proved by Pudl{\'{a}}k \cite{Pud94}. To apply it for super-constant depth, we explicitly determine an appropriate $\beta$ (in terms of $\alpha$ and $\epsilon, \delta$). We reproduce the proof here.

\begin{lemma}
\label{lem:pudlak_depth_reduction}
(Lemma 9 in \cite{Pud94}) Let $f(n) \le \lfloor \sqrt{n} \rfloor$ for every $n$. For every positive reals $\alpha, \beta, \epsilon$ such that if
\begin{equation}
\label{equ:Dlb_d}	
\forall n \forall r \le n \quad D\left(n, d, \frac{\epsilon}{2}, \delta, \frac{1}{r}\right) \,\ge\, \alpha n f(r),
\end{equation}
then
\begin{equation}
\label{equ:Dlb_dp2}
\forall n \forall r \le n \quad D\left(n, d+2, \epsilon, \delta, \frac{1}{r}\right) \,\ge\, \beta n f^*(r),
\end{equation}
where $\beta = \min\left\{\frac{\epsilon\delta}{27}, \alpha\right\}$.
\end{lemma}
\begin{proof} Let $G$ be a depth-$(d+2)$ layered graph with $n$ inputs and $n$ outputs which is $(\frac{\epsilon}{2}, \delta, \frac{1}{r})$-densely regular. Let $V_i$, $i = 0, 1, \ldots, d+2$, be the $i$th level of $G$, where $V_0$ are inputs and $V_{d+2}$ are outputs. Let 
\[
A_0 \,=\, \left\{ v \in V_1 \cup V_{d+1} : \deg(v) > f(r) \right\},
\]
and
\[
A_i \,=\, \left\{ v \in V_1 \cup V_{d+1} : \deg(v) \in (f^{(i+1)}(r), f^{(i)}(r)] \right\}
\]
for $i = 1, 2, \ldots, f^{*}(r)$.

\begin{claim}
\label{claim:three_cases}
For every $i$, where $1 \le i \le \frac{f^*(r)}{2}-3$, at least one of the following inequalities holds:
\begin{equation}
\label{equ:case1}
|A_0 \cup \cdots \cup A_{i-1}| \,\ge\, \frac{\epsilon}{4} \cdot \frac{n}{f^{(i+1)}(r)};
\end{equation}
\begin{equation}
\label{equ:case2}
|\{(u, v) : (u, v) \text{ incident with } A_i \cup A_{i+1} \cup A_{i+2}\}| \,\ge\, \frac{\epsilon\delta}{4}n;
\end{equation}
\begin{equation}
\label{equ:case3}
|\{(u, v) : (u, v) \text{ not incident with } A_0 \cup \cdots \cup A_{i+2}\}| \,\ge\, \alpha n\frac{f^{(i+2)}(r)}{f^{(i+3)}(r)}n.
\end{equation}
\end{claim}
\begin{proof} (of Claim \ref{claim:three_cases}) The proof is reproduced from \cite{Pud94}. For any $i \in [1, 	\frac{f^*(r)}{2}-3]$, suppose neither \eqref{equ:case1} nor \eqref{equ:case2} is true, we prove \eqref{equ:case3}.

Let $k \in \left[\frac{n}{f^{(i+1)}(r)}, n\right]$ be an arbitrary integer. Let $\mathcal{X}, \mathcal{Y}$ be probability distributions on $k$-element subsets of inputs and outputs respectively satisfying the condition in Definition \ref{def:densely_regular}. Let $X \in \mathcal{X}$, $Y \in \mathcal{Y}$ be random subsets of inputs and outputs of size $k$.

By the definition of densely regular graphs, there are at least $\epsilon k$ vertex-disjoint paths connecting them. Out of these, there are at most $\frac{\epsilon}{4} \cdot \frac{n}{f^{(i+1)}(r)} \le \frac{\epsilon}{4} \cdot k$ paths passing through $A_0 \cup \cdots \cup A_{i-1}$ by non \eqref{equ:case1}. On \emph{average}, there are at most $\frac{\epsilon}{4} \cdot k$ paths passing through $A_i \cup A_{i+1} \cup A_{i+2}$, since the expected number of vertices of $X \cup Y$ directly connected with with $A_i \cup A_{i+1} \cup A_{i+2}$ is at most $\frac{k}{\delta n} \cdot \frac{\epsilon\delta}{4}n= \frac{\epsilon}{4} \cdot k$. So, on average, the number of vertex-disjoint paths from $X$ to $Y$ avoiding $A_0 \cup \cdots \cup A_{i+2}$ is at least $\frac{\epsilon}{2} \cdot k$.

We transform $G$ into a depth-$d$ graph as follows:
\begin{itemize}
	\item Omit the vertices in $V_1$ and $V_{d+1}$.
	\item For each pair of edges $(u, v), (v, w)$, where $v \in A_{i+3} \cup A_{i+4} \cup \cdots$, add the edge $(v, w)$.
\end{itemize}
As such, the size of the new graph, denoted by $G'$, is at most $f^{(i+3)}(r)$ times larger than the size of $G$, since $f^{(i+3)}(r)$ is an upper bound on the degrees of vertices in $A_{i+3} \cup A_{i+4} \cup \cdots$. We have shown that $G'$ is $(\frac{\epsilon}{2}, \delta, \frac{1}{f^{(i+1)}(r)})$-densely regular. By the assumption of the lemma, $G'$ is of size at least
\[
D\left(n, d, \frac{\epsilon}{2}, \delta, \frac{1}{f^{(i+1)}(r)}\right) \,\ge\, \alpha n f(f^{(i+1)}(r)) \,=\, \alpha n f^{(i+2)}(r),
\]
which implies that the number of edges not incident with $A_0 \cup \cdots \cup A_{i+2}$ is at least
\[
\frac{|E(G')|}{f^{(i+3)}(r)} \,\ge\, \alpha n\frac{f^{(i+2)}(r)}{f^{(i+3)}(r)}n,
\]
where the last step is by Lemma \ref{lem:pudlak_f_star} (ii).
\end{proof}

Consider the three cases according to Claim \ref{claim:three_cases}.

\textbf{Case 1:} For \emph{some} $i$, where $1 \le i \le \frac{f^*(r)}{2}-3$, \eqref{equ:case1} holds. Since each vertex in $A_0 \cup \ldots \cup A_{i-1}$ has degree at least $f^{(i)}(r)$, we have
\begin{align*}
|E(G)| \;\;\ge\;\; & \frac{\epsilon}{4} \cdot \frac{n}{f^{(i+1)}(r)} \cdot f^{(i)}(r) \\
\;\;\ge\;\; & \frac{\epsilon}{4} \cdot n \cdot f^{(i+1)}(r) && \text{By Lemma \ref{lem:pudlak_f_star} (ii)}\\
\;\;\ge\;\; & \frac{\epsilon}{4} \cdot n \cdot f^{*}(r) && \text{By Lemma \ref{lem:f_star}}.
\end{align*}

\textbf{Case 2:} For \emph{all} $i$, where $1 \le i \le \frac{f^*(r)}{2}-3$, \eqref{equ:case2} holds. If $f^*(r) \ge 54$, we have
\begin{align*}
|E(G)| \;\;\ge\;\; & \frac{1}{3} \cdot \left(\frac{f^*(r)}{2} - 3\right) \cdot \frac{\epsilon\delta}{4}n \\
\;\;\ge\;\; & \frac{1}{3} \cdot \frac{4}{9} f^*(r) \cdot \frac{\epsilon\delta}{4}n \\
\;\;=\;\; & \frac{1}{27} f^*(r) \epsilon \delta n.
\end{align*}
If $f^*(r) < 54$, we still have $|E(G)| \ge \frac{1}{27} \cdot f^*(r) \epsilon \delta n$ by Claim \ref{claim:dr_basic_lb}.

\begin{claim}
\label{claim:dr_basic_lb}
Let $G$ be a layered graph with $n$ inputs and $n$ outputs and of depth at least 2, which is $(\epsilon, \delta, \frac{1}{r})$-densely regular. Then $|E(G)| \ge 2\epsilon \delta n$.	
\end{claim}
\begin{proof} (of the claim) Let $k \in [\frac{n}{r}, n]$. Denote the inputs of $G$ by $I(G)$, and the outputs of $G$ by $O(G)$. Let $\Delta_G(X, Y)$ denote the number of vertex-disjoint paths between $X$ and $Y$.

By the definition of densely regular graphs, there are nonempty sets of $k$-element subsets $\mathcal{X} \in {I(G) \choose k}$ and $\mathcal{Y} \in {O(G) \choose k}$ such that for every $i \in I(G)$, for every $j \in O(G)$,
$
\Pr_{X \in \mathcal{X}}[i \in X] \le \frac{k}{\delta n} \text{ and } \Pr_{Y \in \mathcal{Y}}[j \in Y] \le \frac{k}{\delta n}
$, and $\mathbb{E}[\Delta_G(X, Y)] \ge \epsilon k$.
On the other hand,
\begin{align*}
\mathop{\mathbb{E}}_{X \in \mathcal{X} \atop Y \in \mathcal{Y}}\left[\Delta_G(X, Y)\right] \;\;\le\;\; & \mathop{\mathbb{E}}_{X \in \mathcal{X}}[\sum_{x \in \mathcal{X}} \mathbbm{1}_{\deg(x) > 0}] 
\\
\;\;\le\;\; & \sum_{x \in I(G) \text{ and } \deg(x) > 0} \Pr_{x \in \mathcal{X}}[x \in \mathcal{X}] \\
\;\;\le\;\; & \frac{k}{\delta n} \sum_{x \in I(G)} \deg(x).
\end{align*}
Putting them together, we have $\sum_{x \in I(G)} \deg(x) \ge \epsilon \delta n$. Similarly, we can prove $\sum_{y \in O(G)} \deg(y) \ge \epsilon \delta n$. Since $G$ is a layered graph of depth at least 2, we have $|E(G)| \ge 2 \epsilon \delta n$.
\end{proof}

\textbf{Case 3:} For \emph{some} $i$, where $1 \le i \le \frac{f^*(r)}{2}-3$, \eqref{equ:case3} holds. So we have
\[
|E(G)| \,\ge\, \alpha n\frac{f^{(i+2)}(r)}{f^{(i+3)}(r)}n \,\ge\, \alpha f^{(i+3)}(r) \,\ge\, \alpha n f^*(r).
\]

Combining 3 cases, we have $\beta = \min\left\{\frac{\epsilon}{4}, \frac{\epsilon \delta}{27}, \alpha\right\} = \min\left\{\frac{\epsilon \delta}{27}, \alpha\right\}$.
\end{proof}

We are ready to prove Theorem \ref{thm:pudlak_lb_refinement}.

\begin{proof} (of Theorem \ref{thm:pudlak_lb_refinement}) Since $\lambda_{2i}(n) = \Theta( \lambda_{2i+1}(n) )$ for $i \ge 2$, it suffices to prove
\begin{equation}
\label{equ:pudlak_induction}
D\left(n, 2d+1, \epsilon, \delta, \frac{1}{r}\right) \,\ge\, \min\left\{ \frac{1}{27}, \frac{\delta}{2} \right\} 2^{-d} \epsilon \delta n \lambda_{2d+1}(r).
\end{equation}

The base case for $d = 0$ is proved in Lemma \ref{lem:pudlak_d1}. Assume \eqref{equ:pudlak_induction} is true for $d$, let us prove it for $d + 1$. By induction hypothesis, we have
\[
D\left(n, 2d+1, \frac{\epsilon}{2}, \delta, \frac{1}{r}\right) \,\ge\, \min\left\{ \frac{1}{27}, \frac{\delta}{2} \right\} 2^{-(d+1)} \epsilon \delta n \lambda_{2d+1}(r)
\]
holds for any $r \le n$. By Lemma \ref{lem:pudlak_depth_reduction}, we have
$
D(n, 2d+3, \epsilon, \delta, \frac{1}{r}) \ge \beta n \lambda_{2d+1}^*(r) = \beta n \lambda_{2d+3}(r) ,
$
where 
\begin{eqnarray*}
\beta & = & \min\left\{ \frac{\epsilon \delta}{27}, \min\left\{ \frac{1}{27}, \frac{\delta}{2} \right\} 2^{-(d+1)} \epsilon \delta \right\} \\
& = & \min\left\{\frac{1}{27}, \frac{\delta}{2} \right\} 2^{-(d+1)} \epsilon \delta,
\end{eqnarray*}
which completes the induction step.
\end{proof}


\end{document}